\documentclass{amsart}
\usepackage{soul,times,amssymb,amsthm,amsmath,amsfonts,bbm,mathrsfs,xcolor,subfigure,graphicx,enumitem,booktabs,euscript}
\allowdisplaybreaks
\usepackage{graphicx}
\usepackage[unicode]{hyperref}
\hypersetup{
    colorlinks,
    linkcolor={blue!80!black},
    citecolor={blue!80!black},
    urlcolor={blue!80!black}
}
\usepackage[normalem]{ulem}
\usepackage[font=footnotesize]{caption}
\usepackage[utf8]{inputenc}
\usepackage[T1]{fontenc}

\newtheorem{theorem}{Theorem}[section]
\newtheorem{lemma}[theorem]{Lemma}

\newtheorem{proposition}[theorem]{Proposition}
\newtheorem{corollary}[theorem]{Corollary}
\theoremstyle{remark}
\newtheorem{remark}{Remark}
\newtheorem{definition}{Definition}[section]

\usepackage[top=1.55in,bottom=1.5in,left=1.3in,right=1.3in]{geometry}

\newcommand\nc\newcommand
\nc\bfa{{\boldsymbol a}}\nc\bfA{{\boldsymbol A}}\nc\cA{{\EuScript A}}
\nc\bfb{{\boldsymbol b}}\nc\bfB{{\boldsymbol B}}\nc\cB{{\EuScript B}}
\nc\bfc{{\boldsymbol c}}\nc\bfC{{\boldsymbol C}}\nc\cC{{\mathscr C}}
\nc\bfd{{\boldsymbol d}}\nc\bfD{{\boldsymbol D}}\nc\cD{{\EuScript D}}
\nc\bfe{{\boldsymbol e}}\nc\bfE{{\boldsymbol E}}\nc\cE{{\EuScript E}}
\nc\bff{{\boldsymbol f}}\nc\bfF{{\boldsymbol F}}\nc\cF{{\mathscr F}}
\nc\bfg{{\boldsymbol g}}\nc\bfG{{\boldsymbol G}}\nc\cG{{\EuScript G}}
\nc\bfh{{\boldsymbol h}}\nc\bfH{{\boldsymbol H}}\nc\cH{{\mathcal H}}
\nc\bfi{{\boldsymbol i}}\nc\bfI{{\boldsymbol I}}\nc\cI{{\mathcal I}}
\nc\bfj{{\boldsymbol j}}\nc\bfJ{{\boldsymbol J}}\nc\cJ{{\EuScript J}}
\nc\bfk{{\boldsymbol k}}\nc\bfK{{\boldsymbol K}}\nc\cK{{\EuScript K}}
\nc\bfl{{\boldsymbol l}}\nc\bfL{{\boldsymbol L}}\nc\cL{{\EuScript L}}
\nc\bfm{{\boldsymbol m}}\nc\bfM{{\boldsymbol M}}\nc\cM{{\EuScript M}}
\nc\bfn{{\boldsymbol n}}\nc\bfN{{\boldsymbol N}}\nc\cN{{\EuScript N}}
\nc\bfo{{\boldsymbol o}}\nc\bfO{{\boldsymbol O}}\nc\cO{{\EuScript O}}
\nc\bfp{{\boldsymbol p}}\nc\bfP{{\boldsymbol P}}\nc\cP{{\EuScript P}}
\nc\bfq{{\boldsymbol q}}\nc\bfQ{{\boldsymbol Q}}\nc\cQ{{\EuScript Q}}
\nc\bfr{{\boldsymbol r}}\nc\bfR{{\boldsymbol R}}\nc\cR{{\EuScript R}}
\nc\bfs{{\boldsymbol s}}\nc\bfS{{\boldsymbol S}}\nc\cS{{\EuScript S}}
\nc\bft{{\boldsymbol t}}\nc\bfT{{\boldsymbol T}}\nc\cT{{\EuScript T}}
\nc\bfu{{\boldsymbol u}}\nc\bfU{{\boldsymbol U}}\nc\cU{{\EuScript U}}
\nc\bfv{{\boldsymbol v}}\nc\bfV{{\boldsymbol V}}\nc\cV{{\mathscr V}}
\nc\bfw{{\boldsymbol w}}\nc\bfW{{\boldsymbol W}}\nc\cW{{\mathscr W}}
\nc\bfx{{\boldsymbol x}}\nc\bfX{{\boldsymbol X}}\nc\cX{{\EuScript X}}
\nc\bfy{{\boldsymbol y}}\nc\bfY{{\boldsymbol Y}}\nc\cY{{\mathscr Y}}
\nc\bfz{{\boldsymbol z}}\nc\bfZ{{\boldsymbol Z}}\nc\cZ{{\EuScript Z}}
\nc{\1}{{\mathbbm{1}}}
\nc\reals{{\mathbb R}}
\nc{\integers}{{\mathbb Z}}
\nc{\ff}{{\mathbb F}}

\nc{\remove}[1]{}

\DeclareSymbolFont{bbold}{U}{bbold}{m}{n}
\DeclareSymbolFontAlphabet{\mathbbold}{bbold}
\DeclareMathOperator{\Ent}{Ent}
\DeclareMathOperator{\supp}{supp}

\usepackage{xargs}
\usepackage[colorinlistoftodos,prependcaption]{todonotes}

\usepackage{todonotes}

\usepackage[export]{adjustbox}		

\newcommandx{\rednote}[2][1=]{\todo[linecolor=red,backgroundcolor=red!25,bordercolor=red,#1]{#2}}
\newcommandx{\bluenote}[2][1=]{\todo[linecolor=blue,backgroundcolor=blue!25,bordercolor=blue,#1]{#2}}
\newcommandx{\yellownote}[2][1=]{\todo[linecolor=yellow,backgroundcolor=yellow!25,bordercolor=yellow,#1]{#2}}
\newcommandx{\greennote}[2][1=]{\todo[inline,linecolor=olive,backgroundcolor=green!25,bordercolor=olive,#1]{#2}}

\DeclareMathOperator*{\argmax}{\arg\!\max}

\title{Smoothing of binary codes, uniform distributions, and applications}

\author[]{Madhura Pathegama}\thanks{The authors are with the Department of ECE and Institute for Systems Research, University of Maryland, College Park, MD 20817. Emails: \{madhura,abarg\}@umd.edu. Research supported by NSF grants CCF-2110113 (NSF-BSF) and CCF-2104489.}
\author[]{Alexander Barg}

\begin{document}
\maketitle

\begin{abstract}

The action of a noise operator on a code transforms it into a distribution on the respective space. Some common examples from information theory include Bernoulli noise acting on a code in the Hamming space and Gaussian noise acting on a lattice in the Euclidean space. We aim to characterize the cases when the output distribution is close to the uniform distribution on the space, as measured by R{\'e}nyi divergence of order $\alpha \in (1,\infty]$.  A version of this question is known as the channel resolvability problem in information theory, and it has implications for security guarantees in wiretap channels, error correction, discrepancy, worst-to-average case complexity reductions, and many other problems. 

Our work quantifies the requirements for asymptotic uniformity (perfect smoothing) and identifies explicit code families that achieve it under the action of the Bernoulli and ball noise operators on the code. We derive expressions for the minimum rate of codes required to attain
asymptotically perfect smoothing. In proving our results, we leverage recent results from harmonic analysis of functions on the Hamming space. Another result pertains to the use of code families in Wyner's transmission scheme on the binary wiretap channel. We identify explicit families that guarantee strong secrecy when applied in this scheme, showing that nested Reed-Muller codes can transmit messages reliably and securely over a binary symmetric wiretap channel with a positive rate. Finally, we establish a connection between smoothing and error correction in the binary symmetric channel.
\end{abstract}

\section{Introduction}

Many problems of information theory involve the action of a noise operator on a code distribution, transforming it into some other distribution. For instance, one can think of Bernoulli noise acting on a code in the Hamming space or Gaussian noise acting on a lattice in the Euclidean space. We are interested in characterizing the cases when the output distribution is close to the uniform distribution on the space. Versions of this problem have been considered under different names, including resolvability \cite{han1993approximation,hayashi2006general,yu2018renyi}, smoothing \cite{debris2023smoothing,micciancio2007worst}, discrepancy \cite{Skriganov2001,Chen2002}, and entropy of noisy functions \cite{samorodnitsky2016entropy,samorodnitsky2019upper,samorodnitsky2020improved}. 
Direct applications of smoothing include secrecy guarantees in both the binary symmetric wiretap channel \cite{bloch2013strong,hayashi2006general,yu2018renyi} and the Gaussian wiretap channel \cite{belfiore2010secrecy,luzzi2023optimal}, error correction in the binary symmetric channel (BSC) \cite{hkazla2021codes,rao2022criterion}, converse coding theorems of information theory \cite{arimoto1973converse,han1993approximation,polyanskiy2010arimoto,polyanskiy2013empirical}, strong coordination \cite{cover2007capacity,CuffPermuterCover2010, bloch2013strong, Cuff2013, ChouBlochKliewer2018}, secret key generation \cite{chou2015polar, luzzi2023optimal}, and worst-to-average case reductions in cryptography \cite{micciancio2007worst,brakerski2019worst}. Some aspects of this problem also touch upon approximation problems in statistics and machine learning \cite{nietert2021smooth,goldfeld2022limit,goldfeld2022statistical}.

Our main results are formulated for the smoothing in the binary Hamming space $\cH_n$. For  $r:\cH_n\to \reals_0^+ $, and $f:\cH_n\to \reals $ define
  $$
     T_rf(x)=(r\ast f)(x):=\sum_{z\in \cH_n} r(z)f(x-z)
  $$
be the action of $r$ on the functions on the space. We set $r$ to be a probability mass function (pmf) and call the function $T_rf$ the \emph{noisy version} of $f$ with respect to $r$, and refer to $r$ and $T_r$ as a \emph{noise kernel} and a \emph{noise operator} respectively. By \emph{smoothing} $f$ with respect to $r$ we mean applying the noise kernel $r$ to $f$. We often assume that $r(x)$ is a radial kernel, i.e., its value on the argument $x \in \cH_n$ depends only on the Hamming weight of $x$. 

There are several ways to view the smoothing operation. Interpreting it as a shift-invariant linear operator, we note that, from Young's inequality, $\|T_rf\|_\alpha =  \|f*r\|_{\alpha}  \leq \|f\|_{\alpha}, 1\le\alpha\le\infty,$ so smoothing contracts the $\alpha$-norm. Upon applying $T_r$, the noisy version of $f$ becomes ``flatter'', hence the designation ``smoothing''.
Note that if $f$ is a pmf, then $T_rf$ is also a pmf, and so this view allows us to model the effect of communication channels with additive noise.

The class of functions that we consider are (normalized) indicators of subsets (codes) in $\cH_n$. A code $\cC\subset \cH_n$ defines a pmf 
$f_{\cC}=\frac{\1_{\cC}}{|\cC|},$ and thus $T_rf_{\cC}$ can be viewed as a noisy version of the code (we also sometimes call it a 
\emph{noisy distribution}) with respect to the kernel $r$. The main question of interest for us is the proximity of this distribution to $U_n,$ or the ``smoothness'' of the noisy code distributions. To quantify closeness to $U_n$, we use the Kullback-Leibler (KL) and  R{\'e}nyi divergences (equivalently, $L_{\alpha}$ norms), and the smoothness measured in $D_{\alpha}(\cdot\|\cdot)$ is termed as $D_{\alpha}$-smoothness ($L_{\alpha}$-smoothness).

We say that a code is {\em perfectly smoothable} with respect to the noise kernel $r$ if the resultant noisy distribution becomes uniform.
Our main emphasis is on the asymptotic version of perfect smoothing and its implications for some of the 
basic information-theoretic problems. A sequence of codes $(\cC_n)_n$ is asymptotically smoothed by the 
kernel sequence $r_n$ if the distance between $(T_{r_n}f_{\cC_n})$ and $U_n$ approaches 0 as $n$ 
increases. This property is closely related to the more general problem of {\em channel resolvability} 
introduced by Han and Verd{\'u} in \cite{han1993approximation}. Given a discrete memoryless channel 
$\cW(Y|X)$ and a distribution $P_X$, we observe a distribution $P_Y$ on the output of the channel. The 
task of channel resolvability is to find $P_X$ supported on a subset $\cC\subset \cH_n$ that approximates 
$P_Y$ with respect to KL divergence. As shown in \cite{han1993approximation}, there exists a threshold 
value of the rate such that it is impossible to approximate $P_Y$ using codes of lower rate, 
while any output process can be approximated by a well-chosen code of rate larger than the threshold.
Other proximity measures between distributions were considered for this problem in \cite{steinberg1996simulation,liu2016e_,yu2018renyi}. 
Following the setting in \cite{yu2018renyi}, we consider R{\'e}nyi divergences for measuring closeness to uniformity. We call the minimum rate required to achieve perfect asymptotic smoothing the \emph{ $D_{\alpha}$-smoothing capacity} of noise kernels $(r_n)_n$, where the proximity to uniformity is measured by $\alpha$-R{\'e}nyi divergence. In this work, we characterize the $D_{\alpha}$-smoothing capacity of the sequence $(r_n)_n$ using its R{\'e}nyi entropy rate.

\vspace*{.05in}
\emph{Asymptotic smoothing.}
We will limit ourselves to studying smoothing bounds under the action of Bernoulli noise or the ball noise kernels. 
A common approach to deriving bounds on the norm of a noisy function is through hypercontractivity inequalities \cite{ordentlich2018entropy,polyanskiy2019hypercontractivity,Yu2022}. In its basic version, given a code $\cC$ of size $M$, it yields the estimate
\begin{align*}
    {\|T_{\delta}f_{\cC}\|_{\alpha}  \leq \|f_{\cC}\|_{\alpha^\prime} =  M^\frac{1-\alpha^\prime}{\alpha^\prime}2^{-\frac{n}{\alpha^\prime}},}
\end{align*}
where $\beta_\delta$ is the Bernoulli kernel (see Section~\ref{sec: Prelimnaries} for formal definitions) and $\alpha^\prime = 1+(1-2\delta)^2(\alpha -1)$.
This upper bound does not differentiate codes yielding higher or lower smoothness, which in many situations may not be sufficiently informative. Note that other tools such as ``Mrs. Gerber's lemma'' \cite{wyner1973theorem,ordentlich2018entropy} or strong data processing inequalities also suffer from the same limitation.

A new perspective of the bounds for smoothing has been recently introduced in the works of Samorodnitsky 
\cite{samorodnitsky2016entropy,samorodnitsky2019upper,samorodnitsky2020improved}.
Essentially,  his results imply that codes satisfying certain regularity conditions have good smoothing properties. Their efficiency is highlighted in recent papers  \cite{hkazla2021codes,hkazla2022optimal}, which leveraged results for code performance on the binary erasure channel (BEC) to prove strong claims about the error correction capabilities of the codes when used on the BSC. 
Using Samorodnitsky's inequalities, we show that duals of some BEC capacity-achieving codes achieve $D_{\alpha}$-smoothing capacity for $\alpha \in \{2,3,\dots,\infty\}$ with respect to Bernoulli noise. This includes duals of polar codes and doubly transitive codes such as Reed-Muller (RM) codes.

\emph{Smoothing and the wiretap channel.}
Wyner's wiretap channel \cite{wyner1975wire} models communication in the presence of an eavesdropper. 
Code design for this channel pursues reliable communication between the legitimate parties, while at the same time 
leaking as little information as possible about the transmitted messages to the eavesdropper. 
The connection between secrecy in wiretap channels and resolvability was first mentioned by Csisz{\'a}r \cite{csiszar1996almost} and later developed by Hayashi \cite{hayashi2006general}. It rests on the observation that to achieve secrecy it suffices to make the distribution of eavesdropper's observations conditioned on the transmitted message nearly independent of the message. The idea of characterizing secrecy based on smoothness works irrespective of the measure of secrecy \cite{hayashi2006general,bloch2013strong,yu2018renyi}, and it was also employed for nested lattice codes used over the Gaussian wiretap channel in \cite{belfiore2010secrecy}.

Secrecy on the wiretap channel can be defined in two ways, measured by the information gained by the eavesdropper, and it depends on whether this quantity is normalized to the number of channel uses (weak secrecy) or not (strong secrecy). This distinction
was first highlighted by Maurer \cite{maurer1993secret}, and it has been adopted widely in recent literature. 
Early papers devoted to code design for the wiretap channel relied on random codes, but (for simple channel models such as BSC or 
BEC) this has changed with the advent of explicit capacity-approaching code families. Weak secrecy results based on LDPC codes were 
presented in \cite{thangaraj2007applications}, but initial attempts to attain strong secrecy encountered some obstacles. 
To circumvent it, first works on code construction \cite{Subramanian2010,mahdavifar2011achieving} had to assume that the main channel
is noiseless. The problem of combining strong secrecy and reliability for general wiretap channels was resolved in 
\cite{gulcu2016achieving}, but that work had to assume that the two communicating parties share a small number of random bits unavailable to the eavesdropper. Apart from the polar coding scheme of \cite{gulcu2016achieving}, explicit code families that 
support reliable communication with positive rate and strong secrecy have not previously appeared in the literature. In this work, we show that nested RM codes perform well in binary symmetric wiretap channels based on their smoothing 
properties. While our work falls short of proving that nested RM codes achieve capacity, we show 
that they can transmit messages reliably and secretly at rates close to capacity.  

\vspace*{.05in}\emph{Ball noise and decoding error.} Ball-noise smoothing provides a tool for estimating the error probability of decoding on the BSC. We derive impossibility and achievability bounds for $D_{\alpha}$-smoothness of noisy distributions with respect to the ball noise. Smoothing of a code with respect to the $L_2$ norm plays a special role because in this case the second norm (the variance) of the resulting distribution can be expressed via pairwise distance between codewords, enabling one to rely on tools from 
Fourier analysis. The recent paper by Debris-Alazard \textit{et al.} \cite{debris2023smoothing} established universal bounds for smoothing of codes or lattices, with cryptographic reductions in mind. The paper by Sprumont and Rao 
\cite{rao2022criterion} addressed bounds for error probability of list decoding at rates above BSC capacity. A paper by one of the present authors \cite{barg2021stolarsky} studied the variance of the number of codewords in balls of different radii (a quantity known as quadratic discrepancy \cite{Bilyk2018,Skriganov2019}).

\vspace*{.05in} The main contributions of this paper are the following:
\begin{enumerate}
    \item Characterizing $D_{\alpha}$-smoothing capacities of radial noise operators on the Hamming space for {$\alpha \in (1,\infty]$.}
    \item Identifying some explicit code families that attain smoothing capacity of Bernoulli noise for $\alpha \in \{2,3,\dots,\infty\}$; 
    \item Obtaining rate estimates for RM codes used on the BSC wiretap channel under the strong secrecy condition;
    \item Showing that codes possessing sufficiently good smoothing properties are suitable for error correction.
\end{enumerate}

In Section \ref{sec: Prelimnaries}, we set up the notation and introduce the relevant basic concepts. Then, in Section \ref{sec: Perfect Smoothing Asymptotic}, we derive expressions for $D_{\alpha}$-smoothing capacities for $\alpha \in (1,\infty]$, and in Section \ref{sec: Bernoulli} we use these results to analyze smoothing of code families under the action of Bernoulli noise. 
Section \ref{sec: binary symmetric wiretap  channel} is devoted to the application of these results for the binary symmetric wiretap channel. In particular, we show that RM codes can achieve rates close to the capacity of the BSC wiretap channel, while at the same time guaranteeing strong secrecy. 
In Section~\ref{sec: ball noise} we establish threshold rates for smoothing under ball noise, and derive bounds for the error probability of decoding on the BSC, including the list case, based on the distance distribution. Concluding the paper, Section \ref{sec: Perfect smoothing -finite} briefly points out that the well-known class of uniformly packed codes are perfectly smoothable with respect to ``small'' noise kernels.

\section{Preliminaries}\label{sec: Prelimnaries}

\subsection{Notation}
Throughout this paper, $\cH_n$ is the binary $n$-dimensional Hamming space

{\em Balls and spheres.}
Denote by $B(x,t):= \{y\in \cH_n: |y-x| \leq t\}$ the metric ball of radius $t$ in $\cH_n$ with center at $x$, and denote by $S(x,t):=\{y\in \cH_n: |y-x| = t\}$ the sphere of radius $t$.
Let $V_t=|B(x,t)|$ be the volume of the ball, and let $\mu_t(i) $ be the 
intersection volume of two balls of radius $t$ whose centers are distance $i$ apart:
\begin{align}\label{eq:mu_t}
    \mu_t(i)=  |B(0,t)\cap B(x,t)|, \quad \text{where } |x| = i.
\end{align}

{\em Codes and distributions.} A code $\cC$ is a subset in $\cH_n$. The rate and distance of the code are denoted by $R(\cC):=\log|\cC|/n$ and $d(\cC)$, respectively. Let 
\begin{equation}\label{eq:dd}
    A_i=\frac{1}{|\cC|}|\{(x,y)\in\cC^2: d(x,y)=i\}|
\end{equation}    
and let $(A_i,i=0,\dots,n)$ be the distance distribution of the code. If the code $\cC$ forms an $\ff_2$-linear
subspace in $\cH_n,$ we denote by $\cC^\bot:=\{y\in \cH_n: \sum_i x_iy_i=0 \text{ for all }x\in \cC\}$ its dual code. 

The function $\1_{\cC}$ denotes the indicator of a subset $\cC\subset \cH_n,$ and $f_{\cC} =\frac{\1_{\cC}}{|\cC|}$ is corresponding pmf denoting the uniform distribution over the set, calling it a {\em code distribution}. We use a special notation $s_t$ for this pmf when $\cC=S(0,t)$ and similarly write $b_t$ for it when 
$\cC=B(0,t).$ Finally, $\beta_\delta$ is the binomial distribution on $\cH_n,$ given by
  \begin{equation}\label{eq:beta}
    {\beta}_{\delta}(x) = {\beta}_{\delta}^{(n)}(x) = \delta^{|x|}(1-\delta)^{n-|x|},
  \end{equation}
and $U_n$ is the uniform distribution, given by $U_n(x)=2^{-n}$ for all $x.$

\emph{Entropies and norms.}
For a function $f:\cH_n \rightarrow \mathbbm{R}$, we define its $\alpha${-norm} as follows. 
\begin{align*}
    \|f\|_{\alpha}  &= \Big(\frac{1}{2^n}\sum_{x \in \cH_n}{|f(x)|}^{\alpha} \Big)^{1/{\alpha} } \text{ for } \alpha \in {(0,\infty)}\\
    \|f\|_{\infty} &= \max_{x\in \cH_n}{|f(x)|}.
\end{align*}

Given a pmf $P$, let 
 \begin{align}
   H(P)&=-\sum_i P_i\log P_i \label{eq:Shannon},\\
   H_{\alpha} (P)&=\frac1{1-\alpha }\log\Big(\sum_i P_i^{\alpha} \Big) \label{eq:Renyi}
 \end{align}  
denote its Shannon entropy and R{\'e}nyi entropy of order $\alpha,$ respectively.  If $P$ is supported on two points, we write $h(P)$ and $h_{\alpha} (P)$ instead (all logarithms are to the base 2). The limiting cases of $\alpha=0,1,\infty$ are well-defined; in particular, for $\alpha=1$, $H_\alpha(P)$ reduces to $H(P).$

For two discrete probability distributions $P$ and $\alpha$, the $\alpha$-{\em R{\'e}nyi divergence} (or simply $\alpha$-divergence) is defined as follows:
\begin{equation}
\label{eq: Renyi div}
    D_{\alpha}(P\|Q) = 
\begin{cases}
            - \log Q(\{i:P_i>0\})  & \mbox{if } \alpha = 0 \\[.1in]
		\frac{1}{\alpha -1} \log \sum_i P_i^{\alpha} Q_i^{-(\alpha -1)}  & \mbox{if } \alpha \in (0,1) \cup (1,\infty) \\[.1in]
		\sum_i P_i\log \frac{P_i}{Q_i} & \mbox{if } \alpha =  1\\[.1in]
            \max_i \log \frac{P_i}{Q_i} & \mbox{if } \alpha =  \infty
	\end{cases}.
\end{equation}
The divergence $D_{\alpha}(P\|Q)$ is a continuous function of $\alpha$ for $\alpha \in [0,\infty]$.
For a pmf $f$ on $\cH_n$ 
\begin{align}\label{eq: Renyi smoothness}
    D_{\alpha}(f\|U_n) &=  \frac{\alpha}{\alpha -1}\log \|2^nf\|_{\alpha} , \quad \alpha \in (0,1) \cap (1,\infty)\\
    D_{\infty}(f\|U_n) &= \log\|2^nf\|_{\infty}.\label{eq: Renyi smoothness-inf}
\end{align}

\emph{Channels.} In this paper, a channel is a conditional probability distribution $\cW:\{0,1\}\to \cY,$ where $\cY$ is a finite set, so that $\cW(y|x)$ is the conditional probability of the output $y$ for the input $x$. We frequently consider the binary symmetric channel with crossover probability $\delta$ and the binary erasure channel with erasure  probability $\lambda,$
abbreviating them as ${\text{\rm BSC}}(\delta)$ and ${\text{\rm BEC}}(\lambda),$ respectively. We are often interested in the $n$-fold channel $\cW^{(n)},$ i.e., the conditional probability distribution corresponding to $n$-uses of the channel. For the input $X$, let
$Y_{(X,\cW)}$ be the random output of the channel $\cW^{(n)}$. If the input sequences are chosen from a uniform distribution on a code $\cC$, we denote the input by $X_{\cC}.$ 
Since the number of uses of the channel is usually clear from the context, we suppress the dependency on $n$ from the notation for channels and sequences.

Let $\cC$ be a code of length $n$. For a channel $\cW$ and input $X_{\cC}$, the block-MAP 
decoder is defined as 
\begin{align*}
    \hat{x}(y) = \argmax_{x \in \cC} \Pr(x|y).
\end{align*}
For a given code and channel, denote the error probability of block-MAP decoding by
\begin{align*}
    P_B(\cW, \cC) = \Pr(X_\cC \neq \hat{X}(Y_{(X_{\cC},\cW)}).
\end{align*}

\subsection{\texorpdfstring{$D_{\alpha}$}{}- and \texorpdfstring{$L_{\alpha}$}{}-smoothness}\label{sec:lemmas}

Recall that in the introduction, we expressed the smoothness of a distribution as its proximity to uniformity. Here we formalize this notion based on two (equivalent) proximity measures.

Let $g$ be a pmf on $\cH_n$. A natural measure of the uniformity of $g$ is $D_{\alpha}(g\|U_n)$ ($\alpha \in [0,\infty]$). We call this the $D_{\alpha}$-smoothness of $g$.
Observe that
\begin{align}\label{eq:equiv}
    \|2^ng\|_\alpha &=  \frac{\|g\|_{\alpha} }{\|g\|_1 }\geq 1 \quad \text{ for }  \alpha \in (1,\infty], \text{ and } \\
   \|2^ng\|_\alpha &=  \frac{\|g\|_{\alpha} }{\|g\|_1 }\leq 1 \quad \text{ for }  \alpha \in (0,1)
\end{align}
with equality iff $g= U_n$. Thus, the better the pmf $g$ approximates uniformity, the closer is $\|2^ng\|_{\alpha} $ to 1 (The denominator is simply a normalization quantity that allows dimension-agnostic analysis). Therefore, $\|2^ng\|_{\alpha}$ ($\alpha \in (0,1) \cup (1,\infty]$) can be considered as another measure of proximity. We call $\|2^ng\|_{\alpha}$ the $L_{\alpha}$-smoothness of $g$. From \eqref{eq: Renyi smoothness} and \eqref{eq: Renyi smoothness-inf}, it follows that $D_{\alpha}$-smoothness and $L_\alpha$-smoothness are equivalent. 

\begin{remark}\label{rm: D_alpha mono}
    It is easily seen that $D_{\alpha}(g\|U_n) = n-H_{\alpha}(g)$, hence $D_{\alpha}(g\|U_n)$ is an increasing function of $\alpha$.
\end{remark}

Recall that for a given code $\cC$, and a noise kernel $r$, $T_rf_\cC = r \ast f_\cC$ is the noisy distribution of code $\cC$ with respect to $r$. We intend to study smoothing properties of such noisy distributions of codes. In particular, we characterize the necessary conditions for $D_{\alpha}(T_rf_\cC\|U_n) $ to be close to zero (equivalently, for $\|2^n T_rf_\cC\|_{\alpha}$ close to one). In Section \ref{sec: Perfect Smoothing Asymptotic}, we quantify these requirements in the asymptotic setting.

\subsection{Resolvability}
The problem of channel resolvability was introduced by Han and Verd{\'u}  \cite{han1993approximation} under the name of approximating output statistics of the channel. 
The objective of channel resolvability is to approximate the output distribution of a given input  by the output distribution of a code with a smaller support size. 
In this work, we are interested in code families whose noisy distributions approximate uniformity. Resolvability characterizes the necessary conditions for this to happen in terms of the rate of the code. 

Let $\cW$ be a (discrete memoryless) channel whose input alphabet is $\cX$ and output alphabet is $\cY$. Let $\bfX=\{X_n\}_{n=1}^{\infty}$ be a discrete-time random process where the RVs $X_n$ take
values in $\cX$. Denote by $Y_n$ the random output of $\cW$ with input $X_n$ and
let $\bfY= \{Y_n\}_{n=1}^{\infty}.$ 
Denote by $P_\bfY$ the distribution of $\bfY$ and let $P_{Y^{(n)}}$ be
the pmf of the $n$-tuple $Y^{(n)}:= \{Y_1,Y_2,\dots,Y_n\}$.

For a legitimate (realizable) output process $\bfY$, define 
\begin{align}\label{eq:JDelta}
    J^{(\Delta)}(\cW,P_{\bfY}) = \inf_{\cC_n \subset \cX^n} \{\liminf_{n} R(\cC_n): \Delta(f_{\cC_n}, P_{Y^{(n)}}) \rightarrow 0 \},
\end{align}
where $\Delta$ is a measure of closeness of a pair of probability distributions. 
In words, we look for sequences of distributions $(f_{\cC_n})_n$ of the smallest possible rate that approximates $P_{\bfY}$ on the output of $\cW.$ 

The original problem as formulated by Han and Verd{\'u} in \cite{han1993approximation} seeks to find {\em resolvability} of the channel, defined as
\begin{align}\label{eq:resolvability}
    C_r^{(\Delta)}(\cW) = \inf_{P_{\bfY}}\{J^{(\Delta)}(\cW,P_{\bfY}): {\bfY} \text{ is an output process over } \cW \}.
\end{align}
where $\Delta$ is either the variational distance or normalized KL divergence $\frac1nD(\cdot\|\cdot).$ Hayashi \cite{hayashi2006general} considered the same problem where the proximity was measured by unnormalized KL divergence. In each case, resolvability equals Shannon capacity of the channel $\cW.$

\begin{theorem}[\cite{han1993approximation,hayashi2006general}]\label{thm: Resolvability thm}
Let $\cW$ be a discrete memoryless channel. 
Suppose that $\Delta$ is either the KL divergence (normalized or not) or the variational distance, then resolvability is given by    
    \begin{align*}
        C_r^{(\Delta)}(\cW) = C(\cW).
    \end{align*}
\end{theorem}
The authors of \cite{han1993approximation} proved this result under the additional assumption that the channel $\cW$ satisfies strong converse, and Hayashi \cite{hayashi2006general} later showed that this assumption is unessential.

In addition to the proximity measures considered in Theorem~\ref{thm: Resolvability thm}, papers  \cite{steinberg1996simulation, liu2016e_, yu2018renyi} considered other possibilities. In particular, Yu and Tan \cite{yu2018renyi} studied the resolvability problem for a specific target distribution $P_{\bfY}$ and for the R{\'e}nyi divergence $\Delta=D_{\alpha}$ \eqref{eq: Renyi div}. 
Their main result is as follows. 

\begin{theorem}[\cite{yu2018renyi}, Theorem~2]\label{thm: Renyi resolvability}
Let $\cW$ be a channel and $P_{\bfY}$ be an output distribution. then
\begin{align*}
    J^{(D_\alpha)} (\cW,P_{\bfY}) = 
\left\{
	\begin{array}{lll}
		\min_{P_X \in \cP (\cW,P_{\bfY})}\sum_xP_X(x)D_{\alpha}(\cW(\cdot|x)\|P_{\bfY}) & \mbox{if } \alpha \in (1,2] \cup \{\infty\} \\[.05in]
		\min_{P_X \in \cP (\cW,P_{\bfY})}D(\cW\|P_{\bfY}|P_X) & \mbox{if } \alpha \in (0,1]\\[.05in]
            0 & \mbox{if } \alpha =  0,
	\end{array}
\right.
\end{align*}  
where $\cP (\cW,P_{\bfY})$ is the set of distributions $P_{\bfX}$ consistent with the output $P_{\bfY}$.
\end{theorem}

A direct corollary of Theorem \ref{thm: Renyi resolvability} is the following.

\begin{corollary}[\cite{yu2018renyi},  Eq.~(55)]\label{cor: Renyi resolvability in BSC}
Let $\bfY^{\ast}$ be the output process where for each n, ${Y}_{n}^{\ast}\sim {\text{\rm Ber}}(1/2)$. Then
    \begin{align*}
        J^{(D_\alpha)}  ({\text{\rm BSC}}(\delta),P_{\bfY^{\ast}}) =
        \left\{
	\begin{array}{lll}
		1-h_{\alpha} (\delta) & \mbox{if } \alpha \in (1,2] \cup \{\infty\} \\
		1-h(\delta) & \mbox{if } \alpha \in (0,1]\\
            0 & \mbox{if } \alpha =  0.
	\end{array}
\right.
    \end{align*}
\end{corollary}
This corollary gives necessary conditions for the rate of codes that can approximate the 
uniform distribution via smoothing. We will connect this result to the problem of finding smoothing thresholds in {Section \ref{sec: Bernoulli}}.

\section{Perfect Smoothing -- The Asymptotic case}\label{sec: Perfect Smoothing Asymptotic}
For a given family of noise kernels $(T_{r_n})_n$, there exists a threshold rate such that it is impossible to approximate uniformity with codes of rate below the threshold irrespective of the chosen code, while at the same time, there exist families of codes 
with rate above the threshold that allows perfect approximation in the limit of infinite length. 
For instance, for the Bernoulli$(\delta)$ noise applied to a code $\cC$, the smoothed distribution is nonuniform unless $\cC=\cH_n$ or $\delta=1/2$. At the same time, it is possible to approach the uniform distribution asymptotically for large $n$ once the code sequence satisfies certain conditions. Intuitively it is clear that, for a fixed noise kernel, it is easier to approximate
uniformity if the code rate is sufficiently high. In this section, we characterize the threshold rate for (asymptotically) perfect
smoothing.  Of course, the threshold also depends on the proximity measure $\Delta$ that we are using. In this section, we use perfect smoothing to mean ``asymptotically perfect''. If the proximity measure $\Delta$ for smoothing is not specified, this means that we are using KL divergence.  We obtain threshold rates for perfect smoothing measured with respect to $\alpha$-divergence for several values of $\alpha$. In the subsequent sections, we work out the details for the Bernoulli and ball noise operators, which also have some implications for
communication problems.

\begin{definition}\label{def: asymptotic perfect uniformity-divergence }
Let $(\cC_n)_n$ be a sequence of codes of increasing length $n$ and let $0\le\alpha\le\infty.$
We say that the sequence $\cC_n$ is asymptotically perfectly $D_\alpha$-smoothable with respect to the noise kernels $r_n$  if  $$\lim_{n\to\infty}D_{\alpha}(T_{r_n}f_{\cC_n}\|U_n)= 0,$$ or equivalently \eqref{eq: Renyi smoothness}-\eqref{eq: Renyi smoothness-inf}
if $$\lim_{n\to\infty} \|2^nT_{r_n}f_{\cC_n}\|_{\alpha}= 1 \quad (\alpha\ne 0,1).$$
\end{definition}

One can also define a dimensionless measure for perfect asymptotic smoothing by considering the limiting process
    \begin{equation}\label{eq:pas}
\frac{\|T_{r_n}f_{\cC_n}-U_n\|_{\alpha} }{\|T_{r_n}f_{\cC_n}\|_1 } \rightarrow 0.
  \end{equation}

 \begin{proposition}\label{prop: pas equivalence} Convergence in \eqref{eq:pas} implies perfect smoothing for all $1<\alpha\le\infty$ 
and is equivalent to it for $\alpha\ne \infty.$
\end{proposition}
\begin{proof} 
Since by the triangle inequality, 
   $$
   2^n\|T_rf_\cC\|_\alpha -1 \leq 2^n\|T_rf_\cC-U_n\|_{\alpha} ,
   $$
\eqref{eq:pas} is not weaker than the mode of convergence in
Definition~\ref{def: asymptotic perfect uniformity-divergence } for all $\alpha \in [1,\infty]$.
For $\alpha\ne 1,\infty$ we use Clarkson's inequalities \cite[p.~388]{Simon2015}. Their form
depends on $\alpha$; namely, for 
$2\le\alpha<\infty$ we have
\begin{align*}
    1 + \left\|\frac{2^nT_rf_\cC -1}{2}\right\|_{\alpha}^{\alpha}  &\leq \left\|\frac{2^nT_rf_\cC +1}{2}\right\|_{\alpha}^{\alpha} + \left\|\frac{2^nT_rf_\cC -1}{2}\right\|_{\alpha}^{\alpha}\\
    &\leq \frac{1}{2}(\|2^nT_rf_\cC\|_{\alpha}^{\alpha} + 1).
\end{align*}
For $1<\alpha<2$ the inequality has the form
\begin{align*}
    1 + \left\|\frac{2^nT_rf_\cC -1}{2}\right\|_{\alpha}^{\alpha'}  &\leq \left\|\frac{2^nT_rf_\cC +1}{2}\right\|_{\alpha}^{\alpha'} + \left\|\frac{2^nT_rf_\cC -1}{2}\right\|_{\alpha}^{\alpha'} \\
    &\leq \Big(\frac{1}{2}(\|2^nT_rf_\cC\|_{\alpha}^{\alpha} + 1)\Big)^{\alpha'/\alpha},
\end{align*}
where $\alpha' = \frac{\alpha}{\alpha-1}$ is the H{\"o}lder conjugate. These equations show that for $\alpha \in (1,\infty)$, $ \|2^nT_{r_n}f_{\cC_n}\|_{\alpha} \to 1$ implies $\|2^nT_{r_n}f_{\cC_n}-1\|_{\alpha} \to 0,$ establishing the claimed equivalence.
\end{proof}

\begin{definition}
Let $(r_n)_n$ be a sequence of noise kernels. We say that rate $R$ is achievable for perfect 
$D_\alpha$-smoothing if there exists a sequence of codes 
$(\cC_n)_n$ such that $R(\cC_n)\rightarrow R$ as $n \rightarrow \infty$ and $(\cC_n)_n$ is perfectly $D_\alpha$-smoothable.
\end{definition}

Note that if $R_1$ is achievable, then any rate $1\ge R_2 > R_1$ is also achievable. Indeed, consider a (linear) code $\cC_1$ of rate $R_1$ that has good smoothing properties. 
Construct $\cC_2$ by taking the union of $2^{n(R_2-R_1)}$ non-overlapping shifts of $\cC_1$. Then the rate of $\cC_2$ is $R_2$ and since each shift has good smoothing properties, the same is true for $\cC_2$.
Therefore, let us define the main concept of this section.
\begin{definition}\label{def:capacity}
Given a sequence of kernels $r=(r_n)_n$, define the {\em $D_{\alpha}$-smoothing capacity} as
   \begin{equation}\label{eq:smooth-p}
   S_{\alpha}^{r}:=\inf_{(\cC_n)_n}\{\liminf_{n \to \infty }R(\cC_n): \lim_{n\to\infty}D_{\alpha}(T_{r_n}f_{\cC_n}\|U_n) =0\}.
   \end{equation}
\end{definition}
Note that this quantity is closely related to resolvability: if, rather than optimizing on the output process in \eqref{eq:resolvability}, we set the output distribution to uniform and take $\Delta=D_\alpha,$ then $S^r_\alpha$ equals
$J^{(D_\alpha)}(\cW,P_{\bfy})$ for the channel $\cW$ given by the noise kernel $r.$ To avoid future confusion, we refer to the capacity of reliable transmission as Shannon's capacity. 




The following lemma provides a lower bound for $D_{\alpha}$-smoothness.

\begin{lemma}\label{lemma: D_q(T_r|U) impos}
    Let $\cC\subset\cH_n$ be a code of size $M=2^{nR}$ and let $r$ be a noise kernel. Then for $\alpha \in [0,\infty]$
    \begin{equation*}
        D_{\alpha}(T_rf_\cC\|U_n) \geq  n(1-R)-H_{\alpha}(r).
    \end{equation*}
\end{lemma}

\begin{proof}
We will first prove that $ \|2^nT_{r}f_\cC\|_{\alpha}^{\alpha} \geq 2^{(\alpha-1)[n(1-R)-H_{\alpha}(r)]} $ for  $\alpha \in (1,\infty)$:
\begin{align*}
    \|2^nT_{r}f_\cC\|_{\alpha}^{\alpha} &= \frac{2^{n\alpha}}{2^n}\sum_{x\in \cH_n}T_{r}f_\cC(x)^{\alpha}\\
    &= 2^{n(\alpha-1)}\sum_{x\in \cH_n}[\sum_{y \in \cH_n}r(y)f_\cC(x-y)]^{\alpha}\\
    &\geq 2^{n(\alpha-1)}\sum_{x\in \cH_n}\sum_{y \in \cH_n}[r(y)f_\cC(x-y)]^{\alpha}\\
    & =  2^{n(\alpha-1)}\sum_{y \in \cH_n}r(y)^{\alpha}\sum_{x\in \cH_n}f_\cC(x-y)^{\alpha}\\
    & =  \frac{2^{n\alpha}}{|\cC|^{(\alpha-1)}}\|r\|_{\alpha}^{\alpha}\\
    & = 2^{(\alpha-1)[n(1-R)-H_{\alpha}(r)]}.
\end{align*}
Together with \eqref{eq: Renyi smoothness}, this implies that the claimed inequality holds for  $\alpha \in (1,\infty)$.

A similar calculation shows that for $\alpha \in (0,1)$, $\|2^nT_{r}f_\cC\|_{\alpha}^{\alpha} \leq 2^{(\alpha-1)[n(1-R)-H_{\alpha}(r)]},$ yielding the claim for $\alpha \in (0,1)$. 
The limiting cases $\alpha =0$, $\alpha =1$, and $\alpha =\infty$ follow by continuity of  $D_{\alpha}$ and $H_{\alpha}$ for all $\alpha\ge 0.$
\end{proof}
Define 
   \begin{equation}\label{eq:pi}
   \pi(\alpha) = \liminf_{n\to\infty}\frac{H_{\alpha}(r_n)}{n}
   \end{equation}
Lemma~\ref{lemma: D_q(T_r|U) impos} shows that it is impossible to achieve perfect $D_{\alpha}$-smoothing if $R < 1-\pi(\alpha).$ A question of interest is whether there exist sequences of codes of $R > 1-\pi(\alpha)$ that achieves perfect $D_{\alpha}$-smoothing. The next theorem shows that this is the case for $\alpha \in (1,\infty]$.

\begin{theorem}\label{thm: General smoothing capacities}

Let $r = (r_n)_n$ be a sequence of noise kernels and let $\alpha \in (1,\infty]$. Then 
\begin{align}\label{eq:sPi}
    S_{\alpha}^{r}= 1-\pi(\alpha).
\end{align} 
\end{theorem}

The proof relies on a random coding argument and is given in Appendix~\ref{appendixB}.
This result will be used below to characterize smoothing capacity of the Bernoulli and ball noise operators.
\begin{remark}
  Equality \eqref{eq:sPi} does not hold in the case $\alpha \in [0,1]$. From Theorem \ref{thm: Bernoulli-smoothing capacities} below, Bernoulli noise does not satisfy \eqref{eq:sPi} for $\alpha \in [0,1)$. To construct a counterexample for $\alpha =1$, consider  the noise kernel that is almost uniform except for one distinguished point, for instance,  $r_n(x) = 2^{-(n+1)}$ for $x\ne 0$ and
  $r_n(0)=\frac12+\frac1{2^{n+1}}.$ Performing the calculations, we then obtain
that $S_1^{r}=1$ while $\pi(1)=\frac12.$
\end{remark}

\begin{remark}
    It is worth noting that $\pi(\alpha)$ is a decreasing function of $\alpha$ for $0\le\alpha\le\infty.$
\end{remark}

\section{Bernoulli noise}\label{sec: Bernoulli}
In this section, we characterize the value $S_\alpha^{\beta_\delta}$ for a range of values of $\alpha$. Then we provide explicit code families that attain the $D_{\alpha}$-smoothing capacities.  

As already mentioned, resolvability for $\beta_\delta$ with respect to $\alpha$-divergence was considered by Yu and Tan 
\cite{yu2018renyi}. Their results, stated in Corollary~\ref{cor: Renyi resolvability in BSC}, yield an expression
for $S_{\alpha}^{\beta_\delta}$ for $\alpha \in [0,2]\cup\{\infty\}$. 
The next theorem summarizes the current knowledge about $S_\alpha^{\beta_\delta},$ 
where the claims for $2<\alpha<\infty$ form new results.

\begin{theorem}\label{thm: Bernoulli-smoothing capacities}
\begin{align}
    S_{\alpha}^{\beta_{\delta}}= 
\left\{
	\begin{array}{lll}
		0             & \mbox{if } \alpha = 0\\[.05in]
		1-h(\delta)   & \mbox{if } \alpha \in (0,1]\\[.05in]
            1-h_{\alpha} (\delta) & \mbox{if } \alpha \in (1,\infty].
	\end{array}
\right.
\end{align} 
\end{theorem}
\begin{proof}  The claims for $\alpha\in[0,1]$ follow from Corollary \ref{cor: Renyi resolvability in BSC}.
The results for $\alpha=(1,\infty]$ follow from Theorem \ref{thm: General smoothing capacities} since $\frac{H_{\alpha}(\beta_{\delta})}{n} = h_{\alpha}(\delta)$.
\end{proof} 
Having quantified smoothing capacities, let us examine code families with strong smoothing properties.
Since $D_1$-smoothing capacity and Shannon capacity coincide, it is natural to speculate that codes that achieve Shannon capacity
when used on the BSC$(\delta)$ would also attain $D_1$-smoothing capacity.
However, the following result demonstrates that capacity-achieving codes do not yield perfect smoothing. For typographical reasons, we abbreviate $T_{\beta_\delta}$ by $T_{\delta}$ from this section onward.

\begin{proposition}\label{prop: Bernoulli smoothing cap.-achieving codes}
    Let $\cC_n$ be a sequence of codes achieving capacity of ${\text{\rm BSC}}(\delta)$. Then 
      $$
      D(T_{\delta}f_{\cC_n}\|U_n)  \rightarrow \infty, D(T_{\delta}f_{\cC_n}\|U_n) = o(n).
      $$
\end{proposition}     
\begin{proof}
    The second part of the statement is  Theorem 2 in \cite{shamai1997empirical}. The first part is obtained as follows. 
    Let $\cC_n$ be a capacity-achieving sequence of codes in ${\text{\rm BSC}}(\delta)$. Then from \cite{polyanskiy2010channel} (Theorem 49), there exists a constant $K>0$ such that $nR(\cC_n) \leq n(1-h(\delta)) - K\sqrt{n}$ for large $n$. Therefore, 
    \begin{align*}
        0 \leq  H(X_{\cC_n}|Y_{{\text{\rm BSC}}(\delta),X}) = n(R(\cC_n)+h(\delta)-1)+ D(T_{\delta}f_{\cC_n}\|U_n),
    \end{align*}
    which implies $D(T_{\delta}f_{\cC_n}\|U_n) \geq K\sqrt{n}$.
\end{proof}

Apart from random codes, only polar codes are known to achieve $D_1$-smoothing capacity. Before stating the formal result, recall that polar codes are formed by applying several iterations of a linear transformation to the input, which results in creating virtual channels for individual bits with Shannon's capacity close to zero or to one, plus a vanishing proportion of intermediate-capacity channels. While by Proposition \ref{prop: Bernoulli smoothing cap.-achieving codes}, that polar codes that achieve the BSC capacity, cannot achieve $D_1$-smoothing capacity, adding some intermediate-bit channels to the 
set of data bits makes this possible. This idea was first introduced in \cite{mahdavifar2011achieving}\footnote{
\label{fn:MV} The authors of \cite{mahdavifar2011achieving} had to include these channels to attain secrecy for the wiretap channel. At the same time, it is this inclusion that did not allow them to also attain transmission reliability. See Sec.~\ref{sec: binary symmetric wiretap  channel} for more details about this issue.} and expressed in terms of resolvability in \cite{bloch2012strong}.

\begin{theorem} [\cite{bloch2012strong}, Proposition 1]\label{thm: Bernoulli smoothing polar codes}
Let $\cW$ be the ${\text{\rm BSC}}(\delta)$ channel and $\cW_n^{(i)}$  be the virtual channels formed after applying $n$ steps of the polarization procedure. 
For $\gamma \in (0,1/2)$, define $\cG_n = \{i \in [n]: C(\cW_n^{(i)}) \geq 2^{-n^\gamma}\}$. Let $\cC_n$ be the polar code corresponding to virtual channels $\cG_n$. Then $D(T_{\delta}f_{\cC_n}\|U_n) \to 0$.
\end{theorem}

 Note that $\lim_{n\to\infty}R(\cC_n) =  \lim_{n\to\infty}\frac{|\cG_n|}{n} = 1-h(\delta)$. Hence, the polar code construction presented above achieves the perfect smoothing threshold with respect to KL divergence. Furthermore, since the convergence in $\alpha$ divergence for $\alpha<1$ is weaker than the convergence in $\alpha =1$, the same polar code sequence is perfectly $D_{\alpha}$-smoothable for $\alpha<1$. Noticing that the smoothing threshold for $\alpha<1$ is $1-h(\delta)$ by Theorem \ref{thm: Bernoulli-smoothing capacities},  we conclude that the above polar code sequence achieves smoothing capacity in $\alpha$-divergence for $\alpha<1$.  

As mentioned earlier, the smoothing properties of code families other than random codes and polar codes, have not been extensively studied.  We show that duals of capacity-achieving codes in the BEC exhibit good smoothing properties using the tools developed in \cite{samorodnitsky2020improved}. As the first step, we establish a connection between the smoothing of a generic linear code and the erasure correction performance of its dual code.

\begin{lemma}\label{lemma: smoothing and erasure}
Let $\cC$ be a linear code and let $X_{\cC^{\bot}}$ be a random uniform codeword of ${\cC^{\bot}}$. Let $Y_{X_{\cC^{\bot}},{\text{\rm BEC}}(\lambda)}$ be the output of the erasure channel ${\text{\rm BEC}}(\lambda)$ for the input $X_{\cC^{\bot}}$.    Then
    \begin{align}\label{eq:EGp}
        D_{\alpha}(T_{\delta}f_{\cC}\|U_n) \leq H(X_{\cC^{\bot}}|Y_{X_{\cC^{\bot}},{\text{\rm BEC}}(\lambda)}),
    \end{align}
where $\lambda  = (1-2\delta)^2$ for $\alpha =1$ and $\lambda = 1-h_{\alpha} (\delta)$ for $\alpha \in \{2,3,\dots,\infty\}$. 
\end{lemma}
The proof is given in Appendix \ref{appendix: proof: lemma: smoothing and erasure}.

Using this lemma, we show that duals of BEC capacity-achieving codes  (with growing distance) exhibit good smoothing properties. In particular, they achieve $D_{\alpha}$-smoothing capacities for $\alpha \in \{2,3,\dots,\infty\}$.

\begin{theorem}\label{thm: Bernoulli Smoothing - BEC capacity-achieving codes}
Let $(\cC_n)_n$ be a sequence of linear codes with rate $R_n\to R$. Suppose that
the dual sequence $(\cC_n^\bot)_n$ achieves Shannon's capacity of the BEC$(\lambda)$ with $\lambda=R$, and assume that $d(\cC_n^\bot) = \omega(\log n).$ If  $R> (1-2\delta)^2$, then 
    \begin{align*}
        \lim_{n \to \infty}D(T_{\delta}f_{\cC_n}\|U_n) = 0. 
    \end{align*}
Additionally, for $\alpha \in \{2,3,\dots,\infty\}$, if $R>1-h_{\alpha} (\delta)$, then 
    \begin{align*}
       \lim_{n \to \infty} D_{\alpha}(T_{\delta}f_{\cC_n}\|U_n) = 0. 
    \end{align*}
    In particular, the sequence $\cC_n$ achieves $D_\alpha$-smoothing capacity $S_\alpha^{\beta_\delta}$ for $\alpha \in \{2,3,\dots,\infty\}$.
\end{theorem}
\begin{proof}
Since the dual codes achieve capacity of the BEC, it follows from \cite[Theorem~ 5.2]{tillich2000discrete} that, if their distance grows with $n$, then their decoding error probability vanishes. In particular,
if   $ d(\cC_n^\bot) = \omega(\log(n)),$ then 
$P_B({\text{\rm BEC}}(R-\epsilon),{\cC_n^{\bot}}) = o(\frac{1}{n})$ for all $\epsilon \in (0,R].$ Hence, from Fano's inequality, $$\lim_{n \to \infty} H(X_{\cC_n^{\bot}}|Y_{X_{\cC_n^{\bot}},{\text{\rm BEC}}(R-\epsilon)}) = 0.$$ 
Now if $R> (1-2\delta)^2$, then there exists $\epsilon_0$ such that $R-\epsilon_0 = (1-2\delta)^2$. Therefore, from Lemma \ref{lemma: smoothing and erasure},  
\begin{align*}
    \lim_{n \rightarrow \infty}D(T_{\delta}f_{\cC_n}\|U_n) \leq  \lim_{n \to \infty} H(X_{\cC_n^{\bot}}|Y_{X_{\cC_n^{\bot}},{\text{\rm BEC}}(R-\epsilon_0)}) = 0.
\end{align*}
Similarly, $\text{ if }  R> 1- h_{\alpha} (\delta) \text{ for } \alpha \in \{2,3,\dots,\infty\}, \text{ then } \lim_{n \rightarrow \infty} D_{\alpha}(T_{\delta}f_{\cC_n}\|U_n) = 0.$

Together with Theorem \ref{thm: Bernoulli-smoothing capacities}, we have now proved the final claim.
\end{proof}

The known code families that achieve capacity of the BEC include polar codes, LDPC codes, and doubly transitive codes, such as constant-rate RM codes. LDPC codes do not fit the assumptions because of low dual distance, but the other codes do. This yields explicit
families of codes that achieve $D_{\alpha}$-smoothing capacity.

We illustrate the results of this section in Fig.~\ref{fig: Bernoulli Smoothing - Thresholds}, where the curves show the achievability and impossibility rates for perfect smoothing with respect to Bernoulli noise. Given a code (sequence) of rate
$R$, putting it through a noise ${\beta_\delta}$ below Shannon capacity cannot achieve
perfect smoothing. The sequence of polar codes from \cite{mahdavifar2011achieving}, cited in Theorem~\ref{thm: Bernoulli smoothing polar codes}, is smoothable at rates
equal to the Shannon capacity (we stress again that they do not provide a decoding
guarantee at that noise level; see footnote~\ref{fn:MV}). At the second
curve from the bottom, the duals of codes that achieve Shannon's capacity in BEC achieve perfect $D_1$-smoothing; at the third (fourth) curve, these codes
are perfectly $D_2$- (or $D_\infty$-) smoothable, and they achieve the corresponding
smoothing capacity.

\begin{figure}[ht]
\hspace*{-.5in}\includegraphics[width=17cm,center]{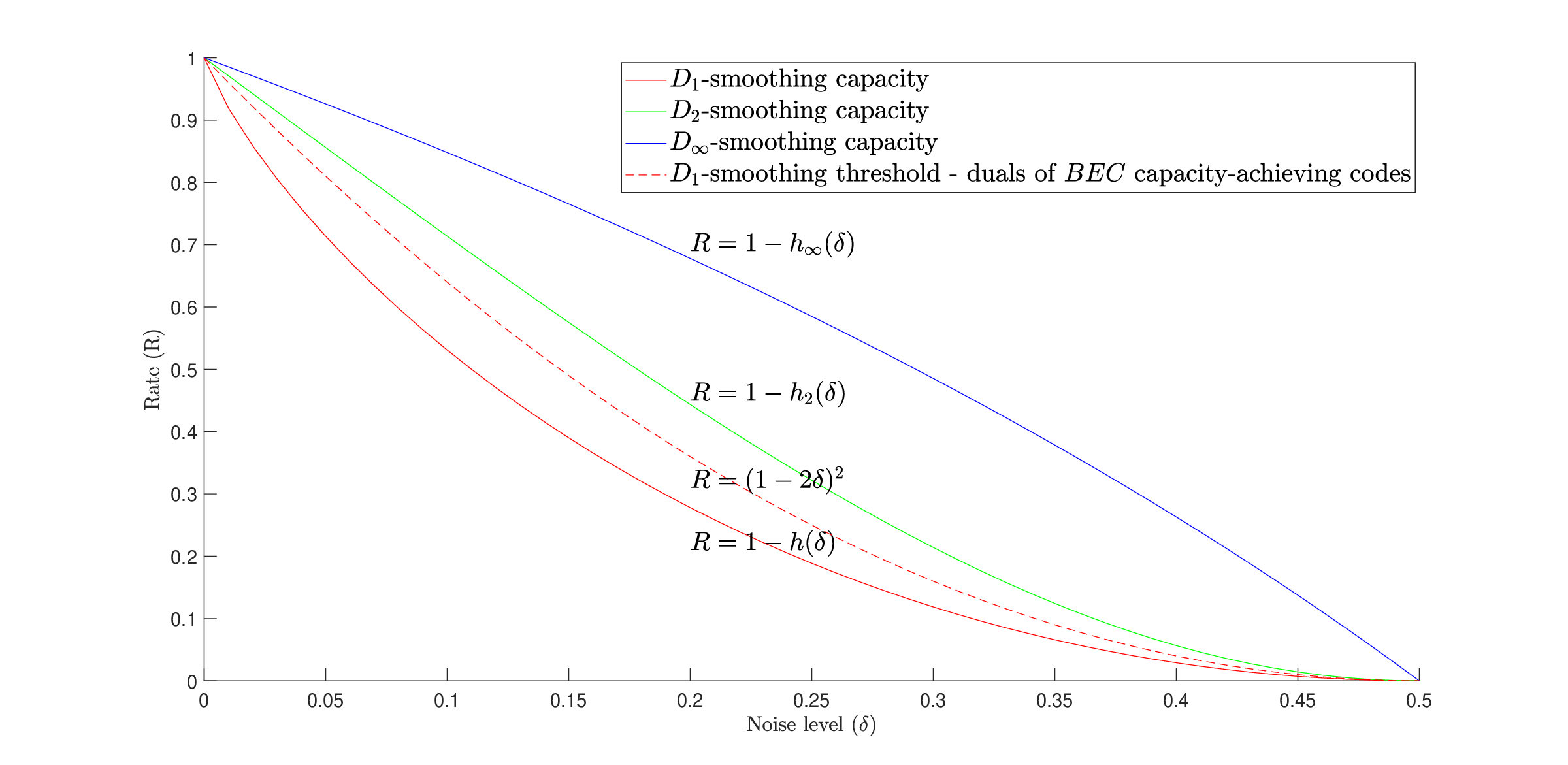}
\centering
\caption{Capacities and achievable rates for perfect smoothing. The lowermost curve gives Shannon capacity of the BSC$(\delta)$, the second curve
from the bottom is the smoothing threshold for duals of BEC capacity-achieving codes, the third one is $S_2^{\beta_\delta}$ and the top one is $S_\infty^{\beta_\delta}.$} 
\label{fig: Bernoulli Smoothing - Thresholds}
\end{figure}

\begin{remark} 
Observe that the strong converse of the channel coding theorem does not imply perfect smoothing. To give
a quick example, consider a code $\cC_n=B(0,\delta'n)$ formed of all vectors in 
the ball. Let $0<\delta<1/2$ and let us use this code on a BSC$(\delta)$, where
$h(\delta)+h(\delta') > 1$ and $\delta<1/2$. 
From the choice of the parameters, the rate of $\cC_n$ is above capacity, and therefore, $P_B({\text{\rm BSC}}(\delta),\cC_n) \approx 1$ from the strong converse.
At the same time,
\begin{align*}
    D(T_{\delta}f_{\cC_n}\|U_n) &= n - H(b_{\delta' n} \ast \beta_{\delta}) =n - H(\beta_{\delta'} \ast \beta_{\delta}) + O(\sqrt{n}) \\
    & = n(1-h(\delta'(1-\delta)+\delta(1-\delta'))) +  O(\sqrt{n}) .
\end{align*}
where the transition from the ball noise to Bernoulli noise (the second equality)
is shown in \cite{ordentlich2018entropy}. Since $\delta'(1-\delta)+\delta(1-\delta'))<1/2$ for all $\delta<1/2, \delta'<1/2,$ 
we conclude that $D(T_{\delta}f_{\cC_n}\|U_n)\nrightarrow 0.$ 
\end{remark}

\begin{remark} In this paper, we mostly study the tradeoff between the rate of codes and the level of the noise needed
to achieve perfect smoothing. A recent work of Debris-Alazard \emph{et al.} \cite{debris2023smoothing} considered
guarantees for smoothing derived from the distance distribution of codes and their dual distance (earlier, similar
calculations were performed in \cite{barg2021stolarsky,ashikhmin2002bounds}). Our approach enables us to 
find conditions for perfect smoothing similar to \cite{debris2023smoothing} but relying on fewer assumptions.
\begin{proposition}\label{prop: sm and dual_dist}
Let $\cC_n$ be a sequence of codes whose dual distance $d(\cC_n^\bot) \geq \partial^\bot n$ where $\partial^\bot \in (0,1)$. If $\partial^\bot > (1-2\delta)^2$, then
    \begin{align*}
        \lim_{n \to \infty} D(T_{\delta}f_{\cC_n}\|U_n) = 0.
    \end{align*}
\end{proposition}
\begin{proof}
    Notice that $\lim_{n \to \infty} H(X_{\cC_n^{\bot}}|Y_{X_{\cC^{\bot}},{\text{\rm BEC}}(\lambda)}) = 0$ if $\partial^\bot > \lambda$. With this, the proof is a straightforward application of Lemma~\ref{lemma: smoothing and erasure}. 
\end{proof}
\end{remark}
Compared to \cite{debris2023smoothing}, this claim removes the restrictions on the support of the dual distance distribution of the codes $\cC_n.$

\section{Binary symmetric wiretap channels }\label{sec: binary symmetric wiretap channel}

In this section, we discuss applications of perfect smoothing to the BSC wiretap 
channel. Wyner's wiretap channel model $\cV$ \cite{wyner1975wire} for the case of 
BSCs is defined as follows. The system is formed of three terminals, $A,B,$ and $E$. 
Terminal $A$ communicates with $B$ by sending messages $M$ chosen from a finite set 
$\cM$. Communication from $A$ to $B$ occurs over a BSC $\cW_b$ with crossover 
probability $\delta_b$, and it is observed by the eavesdropper $E$ via another BSC 
$\cW_e$ with crossover probability $\delta_e>\delta_b.$ A message $M\in\cM$ is 
encoded into a bit sequence $X\in\cH_n$ and sent from $A$ to $B$ in $n$ uses of the 
channel $\cW_b.$ Terminal $B$ observes
the sequence $Y=X+W_b,$ where $W_b\sim\text{Bin}(n,\delta_b)$ is the noise vector, while terminal $E$ observes the sequence $Z=X+W_e$ with $W_e\sim \text{Bin}(n,\delta_e)$. We assume that messages are encoded into a subset of $\cH_n,$ which imposes some probability distribution on the input of the channels.
The goal of the encoding is to ensure reliability and secrecy of communication. The reliability requirement amounts to the condition $\Pr(M \neq \hat{M}) \to 0$ as $n\to\infty$, where $\hat{M}$ is the estimate of $M$ made by $B$. To ensure
secrecy, we require the {\em strong secrecy condition} $I(M;Z) \rightarrow 0$. This is in contrast to the condition
$\frac 1n(M;Z) \rightarrow 0$ studied in the early works on the wiretap channel, which is now called weak secrecy. Denote by $R=\frac 1n\log|\cM|$ the transmission rate. The {\em secrecy capacity} $C_s(\cV)$ is defined as the supremum of the rates that permit reliable transmission, which also conforms to the secrecy condition.

The nested coding scheme, proposed by Wyner \cite{wyner1975wire}, has been 
the principal tool of constructing well-performing transmission protocols for the wiretap channel \cite{thangaraj2007applications,mahdavifar2011achieving,gulcu2016achieving}. To describe it, 
let $\cC_e $ and $\cC_b $ be two linear codes such that $\cC_e \subset \cC_b$ and $|\cM| = \frac{|\cC_b|}{|\cC_e|}$. We assign each message $m$ to a unique coset of $\cC_e$ in $\cC_b$. The sequence transmitted by $A$ is a uniform random vector from 
the coset. As long as the rate of the code $\cC_b$ is below the capacity of $\cW_b$,
we can ensure the reliability of communication from $A$ to $B$.

Strong secrecy can be achieved relying on perfect smoothing. Denote by $c_m$ a leader
of the coset that corresponds to the message $m$. The basic idea is that if $P_{Z|M=m} = (T_{\delta}f_{\cC_e})
(\cdot+c_m)$ is close to a uniform distribution $U_n$ for all $m$, these conditional pmf's are almost indistinguishable from each other, and terminal $E$ has no means of inferring the transmitted message from the observed bit string $Z$. 

As mentioned earlier, weak secrecy results for the wiretap channel based on LDPC codes and on polar codes were presented in \cite{thangaraj2007applications} and \cite{mahdavifar2011achieving}, respectively. The problem that these schemes faced, highlighted in Theorems \ref{prop: Bernoulli smoothing cap.-achieving codes}, \ref{thm: Bernoulli smoothing polar codes}, is that code sequences that achieve BSC capacity have a rate gap of at least $1/\sqrt n$ to the capacity value. At the same time, the rate of perfectly smoothable codes must exceed capacity by a similar quantity \cite{Watanabe2014}. For this reason, the authors of \cite{mahdavifar2011achieving} included the intermediate virtual channels in their polar coding scheme, which gave them strong secrecy, but required a noiseless main channel (cf. footnote~\ref{fn:MV}). A similar general 
issue arose earlier in attempting to use LDPC codes for the wiretap channel  \cite{Subramanian2010}. 

Contributing to the line of work connecting smoothing and the wiretap channel
\cite{hayashi2006general,bloch2013strong,yu2018renyi}, we show that nested coding schemes $\cC_e\subset \cC_b,$ where 
the code $\cC_b$ is good for error correction in ${\text{\rm BSC}}(\delta_b)$ and $\cC_e$ is perfectly smoothable with respect to $\beta_{\delta_b}$, attain strong secrecy and reliability for a BSC wiretap channel $(\delta_b,\delta_e)$.
As observed in Lemma \ref{lemma: smoothing and erasure}, 
duals of good erasure-correcting codes are perfectly smoothable for certain noise levels and hence they form a good choice for $\cC_e$ in this scenario. 

The following lemma establishes a connection between the smoothness of a noisy distribution of a code and strong secrecy.
\begin{lemma}\label{lemma: secrecy and smoothing}
    Consider the nested coding scheme for the BSC wiretap channel introduced above. If \\
    $D(T_{\delta_e}f_{\cC_e}\|U_n) < \epsilon,$ then $I(M;Z) < \epsilon.$
\end{lemma}
\begin{proof}  
 We have    
 \begin{align*}
        D(P_{Z|M}\|U_n|P_M) &= \sum_{\substack{m\in \cM \\z\in \cH_n  }}
        P_{MZ}(m,z)\log\frac{P_{Z|M}(z|m)}{U_n(z)}        
        \\
        &=I(M;Z) + D(P_{Z}\|U_n).
    \end{align*}
 Now note that $P_{Z|M=m}(z) = (T_{\delta_e}f_{\cC_e})(z+c_m) = P_{Z|M=m'}(z+c_{m'}+c_m),$ so $D(P_{Z|M=m}\|U_n) $ is independent of $m$. Therefore, for all  $m \in \mathcal{M}$
    \begin{align*}
        D(P_{Z|M=m}\|U_n) &= D(P_{Z|M}\|U_n|P_M)\\
        &=
        I(M;Z) + D(P_{Z}\|U_n) \geq I(M;Z). \qedhere
    \end{align*}
 \end{proof}
This lemma enables us to formulate conditions for reliable communication while guaranteeing the strong secrecy condition. 
Namely, it suffices to take a pair (a sequence of pairs) of nested codes $\cC_e\subset \cC_b$ such that $D(T_{\delta_e}f_{\cC_e}\|U_n)  \rightarrow 0$ as $n\to\infty.$ If at the same time the code $\cC_b$ corrects errors on a BSC$(\delta_b),$ then the scheme fulfills both the reliability and strong secrecy requirements under noise levels $\delta_b$ and $\delta_e$ for channels $\cW_b$ and $\cW_e$, respectively, supporting transmission from $A$ to $B$ at rate $R_b-R_e$. Together with the 
results established earlier, we can now make this claim more specific. 

\begin{theorem}\label{thm: Wiretap channels achievable rates for BEC capacity-achieving codes}
Let $((\cC_e^{n})^{\bot})_n$ and $(\cC_b^{n})_n$ be sequences of linear codes that achieve capacity of the BEC for their respective rates.
Suppose that $\cC_e^{n}\subset\cC_b^{n}$ and 
\begin{enumerate}
    \item  $d((\cC_e^{n})^\bot) = \omega(\log n), R(\cC_e^{n})\rightarrow R_e$;\\[-.1in]
    \item $d(\cC_b^{n}) = 
\omega(\log n), R(\cC_b^{n})\rightarrow R_b$.
\end{enumerate}
If $R_b 
<1-\log(1+2\sqrt{\delta_b(1-\delta_b)})$ and $R_e > 4\delta_e(1-\delta_e)$, then the nested coding scheme based on 
$\cC_e^{n}$ and $\cC_b^{n}$ can transmit messages with rate $R_b-R_e$ from $A$ to $B$, 
satisfying the reliability and strong secrecy conditions.
\end{theorem}
\begin{proof}
From Corollary \ref{cor: BEC capacity achieving codes in BSC}, the conditions 
$d(\cC_b^{(n)}) = \omega(\log n)$ and $R_b <1-\log(1+2\sqrt{\delta_b(1-\delta_b)})$ 
guarantee transmission reliability. Furthermore, by Theorem \ref{thm: Bernoulli 
Smoothing - BEC capacity-achieving codes}, the conditions $d((\cC_e^{n})^\bot) = 
\omega(\log n)$ and $R_e > 4\delta_e(1-\delta_e)$ imply that $D(T_{\delta_e}f_{\cC_e}
\|U_n)  \rightarrow 0$, which in its turn implies strong secrecy by Lemma~\ref{lemma: 
secrecy and smoothing}.
\end{proof}

To give an example of a code family that satisfies the assumptions of this theorem, consider RM codes of constant
rate. Namely, let $\cC_e^{n} \subset \cC_b^{n}$ be two sequences of RM codes whose rates converge to $R_e$ and $R_b$ respectively. Note that duals of RM codes are themselves RM codes. By a well-known result \cite{kudekar2016reed}, RM codes achieve capacity of the BEC, and for any sequence of constant-rate RM codes, the distance scales as $2^{\Theta(\sqrt{n})}$. Therefore, RM codes
satisfy assumptions of Theorem~\ref{thm: Wiretap channels achievable rates for BEC capacity-achieving codes}.

Note that for RM codes we can obtain a stronger result, based on their error correction properties on the BSC. Involving this
additional argument brings them closer to secrecy capacity under the strong secrecy assumption.

\begin{theorem}\label{thm: Wiretap RM}
    Let $\cC_e^{n}$ and $\cC_b^{n}$ be two sequences of RM codes satisfying $\cC_e^{n}\subset\cC_b^{n}$ whose rates approach $R_e>0$  and $R_b>0,$ respectively. If $R_b <1-h(\delta_b)$ and $R_e > 4\delta_e(1-\delta_e)$, then the nested coding scheme based on $\cC_e^{n}$ and $\cC_b^{n}$ supports transmission on a BSC wiretap channel $(\delta_b,\delta_e)$ with rate $R_b-R_e,$ 
guaranteeing communication reliability and strong secrecy.
\end{theorem}
\begin{proof}
    Very recently, Abbe and Sandon \cite{abbe2023proof}, building upon the work of Reeves and Pfister \cite{reeves2021reed}, proved that RM codes achieve capacity in symmetric channels. Therefore, the condition $R_b <1-h(\delta_b)$ guarantees reliability. The rest of the proof is similar to that of Theorem \ref{thm: Wiretap channels achievable rates for BEC capacity-achieving codes}.
\end{proof}

Theorems \ref{thm: Wiretap channels achievable rates for BEC capacity-achieving codes}
and \ref{thm: Wiretap RM} stop short of constructing codes that attain secrecy capacity of the channel (this is similar to the results of \cite{hkazla2021codes} for
the transmission problem over the BSC). To quantify the gap to capacity,
we plot the smoothing and decodability rate bounds in Fig.~\ref{fig: Wiretap channel - Thresholds}.

\begin{figure}[ht]
\includegraphics[width=17cm,center]{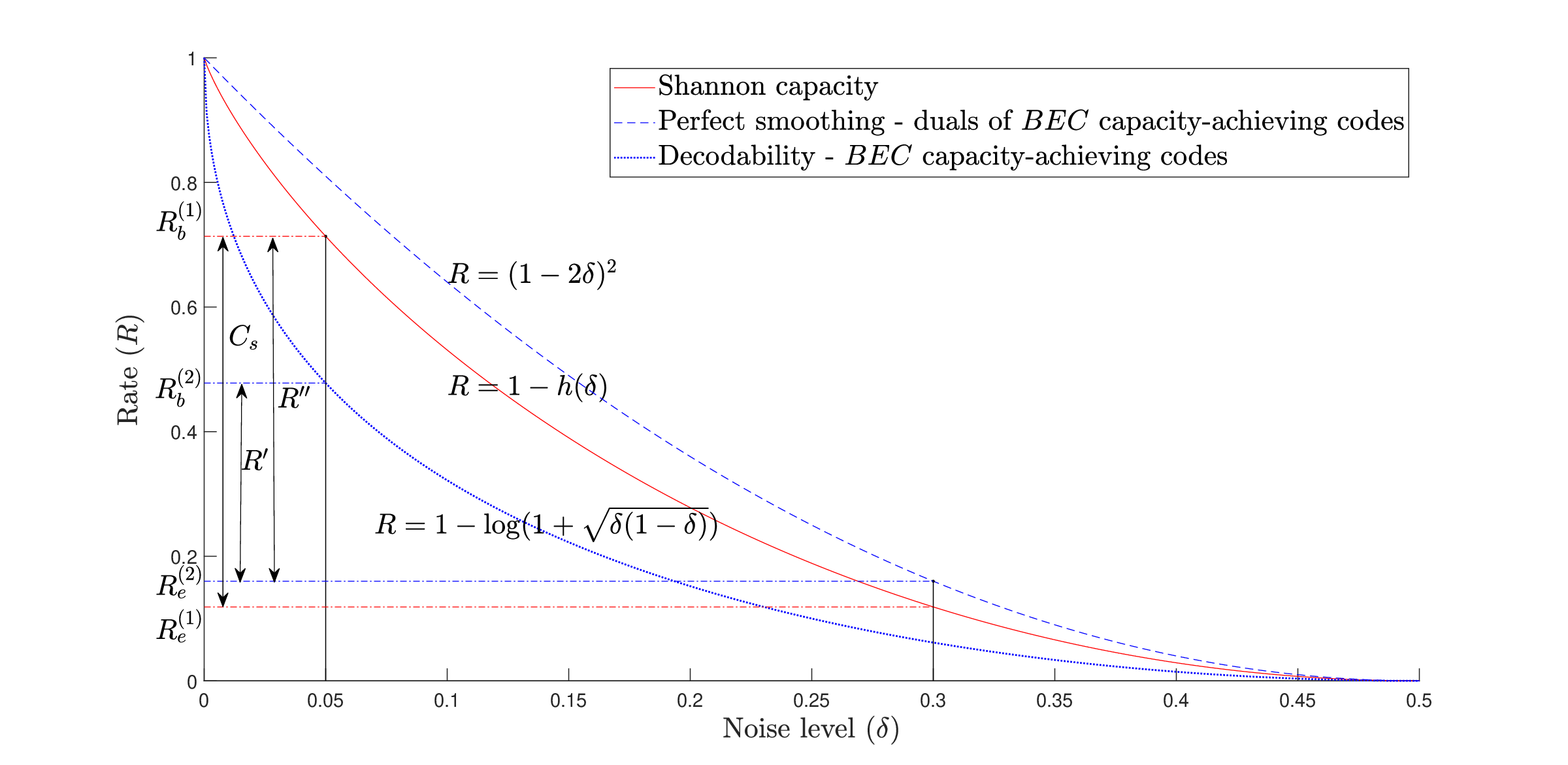}
\caption{Achievable rates in the BSC wiretap channel with BEC capacity-achieving codes. The bottom curve is the lower bound on the code rate that guarantees decodability
on a BSC$(\delta)$. The middle curve shows Shannon's capacity and the top one is the
$D_1$-smoothing threshold for Bernoulli noise $T_{\delta}.$}
\label{fig: Wiretap channel - Thresholds}
\end{figure}

As an example, let us set the noise parameters $\delta_b = 0.05$ and $\delta_e = 0.3$ and denote the 
corresponding secrecy capacity by $C_s$. Suppose that we use a BEC capacity achieving code as code $\cC_b$ 
and a dual of a BEC capacity achieving code as code $\cC_e$ in the nested scheme. The value $R'$ is the largest rate at which
we can guarantee both reliability and strong secrecy.
In the example in Fig.~\ref{fig: Wiretap channel - Thresholds}, $C_s = R_b^{(1)} - R_e^{(1)} = 0.5949$ and $R' = R_b^{(2)} - 
R_e^{(2)} = 0.3181$. The only assumption required here is that the codes $\cC_e^\bot$ and $\cC_b$ have good erasure correction properties.

As noted, generally, RM codes support a higher communication rate than $R'$. Let $R''$ be their achievable rate. For the same noise parameters as above, we obtain $R'' = R_b^{(1)} - R_e^{(2)} = 0.5536,$ which is closer to $C_s$ than $R'$.

\begin{remark}
The fact that RM codes achieve capacity in symmetric channels immediately implies that nested RM codes achieve secrecy 
capacity in BSC wiretap channel under weak secrecy. While it is tempting to assume that, 
coupled with channel duality theorems of \cite{Renes2018,Rengaswamy2021}, this result also implies that RM codes fulfil the strong secrecy requirement on the BSC wiretap channel, an immediate proof looks out of reach \cite{pfister2023private}.
\end{remark}

\subsection{Secrecy from \texorpdfstring{$\alpha$-}{}divergence}
Classically, the (strong) secrecy in the wiretap channel is measured by $I(M,Z)$. In \cite{bloch2013strong}, slightly weaker secrecy measures were considered besides the mutual information. However, more stringent secrecy measures may require in certain scenarios. $\alpha$-divergence-based secrecy measured were introduced by Yu and Tan \cite{yu2018renyi} as a solution to this problem.

Observe that secrecy measured by $D_{\alpha}(P_{Z|M}\|U_n|M)$ for $\alpha \geq 1$ is stronger than the mutual-information-based secrecy. This is because for $\alpha \geq 1$
\begin{align*}
    I(M;Z) \leq D(P_{Z|M}\|U_n|P_M) \leq D_{\alpha}(P_{Z|M}\|U_n|P_M).
\end{align*}

Given a wiretap channel with an encoding-decoding scheme, we say $\alpha$-secrecy is satisfied if 
\begin{align*}
    \lim_{n\to \infty } D_{\alpha}(P_{Z|M}\|U_n|P_M) = 0.
\end{align*}

The following theorem establishes that it is possible to achieve the rate $C(\delta_b)-S_{\alpha}^{\beta_{\delta_e}} = h_{\alpha}(\delta_e) - h(\delta_b)$ with RM codes for $\alpha \in \{2,3,\dots,\infty\}$.

\begin{theorem}
    Let $\alpha \in \{2,3,\dots,\infty\}$.
    Let $\cC_e^{n}$ and $\cC_b^{n}$ be two sequences of RM codes satisfying $\cC_e^{n}\subset\cC_b^{n}$ whose rates approach $R_e>0$  and $R_b>0,$ respectively. If $R_b <1-h(\delta_b)$ and $R_e > 1-h_{\alpha}(\delta_e)$, then the nested coding scheme based on $\cC_e^{n}$ and $\cC_b^{n}$ supports transmission on a BSC wiretap channel $(\delta_b,\delta_e)$ guaranteeing $\alpha$-secrecy with rate $R_b-R_e,$ provided that $h_{\alpha}(\delta_e) - h(\delta_b) > 0$.
\end{theorem}

Evidently, to achieve a stringent version of secrecy, it is necessary to reduce the rate of the message. The capacity of the $(\delta_b,\delta_e)$-wiretap channel is $h(\delta_e) - h(\delta_b)$, while the known highest rate that assures $\alpha$-secrecy and reliability is $h_{\alpha}(\delta_e) - h(\delta_b)$. Hence, to achieve $\alpha$-secrecy, we must give up  $h(\delta_e) - h_{\alpha}(\delta_e)$ of the attainable rate.

\section{Ball noise and error probability of decoding}\label{sec: ball noise}
This section focuses on achieving the best possible smoothing with respect to ball noise. As an application, we show that codes that possess good smoothing properties with respect to ball noise are suitable for error correction in the BSC.

\subsection{Ball noise} Recall that perfect smoothing of a sequence of codes is only possible if the rate is greater than the corresponding $D_{\alpha}$-smoothing capacity. In addition to characterizing $D_{\alpha}$-smoothing capacities of ball noise, we quantify the best smoothing one can expect with rates below the $D_{\alpha}$-smoothing capacity. We will use these results in the upcoming subsection when we derive upper bounds for the decoding error probability on a BSC. The next theorem summarizes our main result on smoothing with respect to ball noise.

\begin{theorem}\label{Thm: ball noise D_alpha sm bounds}
 Let $\left(b_{\delta n}\right)_n$ be the sequence of ball noise operators, where $\delta n$ is the radius of the ball. 
Let $\delta\in[0,1/2],\alpha\in[0,\infty].$ 
Let $\cC_n$ be a code of length $n$ and rate $R_n.$ 
Then we have the following bounds: 
\begin{align}
    D_{\alpha}(T_{b_{\delta n}}f_{\cC_n}\|U_n) &\geq 0      \label{eq: C<R smoothing} \\ 
    \frac{1}{n}D_{\alpha}(T_{b_{\delta n}}f_{\cC_n}\|U_n) &\geq 1-R_n-h(\delta). \label{eq: C>R smoothing}
\end{align}
There exist sequences of codes of rate $R_n\to R$ that achieve asymptotic equality in \eqref{eq: C<R smoothing} for all $R>1-h(\delta).$ 
At the same time, if $R<1-h(\delta)$, then there exist sequences of codes achieving asymptotic equality in \eqref{eq: C>R smoothing}.
\end{theorem}

\begin{proof}
The inequality in \eqref{eq: C<R smoothing} is trivial. Let us prove that asymptotically it can be achieved with equality.
From Theorem \ref{thm: General smoothing capacities}, there exists a sequence of codes  $(\cC_n)_n$ such that $D_{\infty}(T_{b_{\delta n}}f_{\cC_n}\|U_n) = o(1)$ given that $R>1-h(\delta)$. Hence, for $\alpha \in [0,\infty]$
\begin{align*}
    0 \leq D_{\alpha}(T_{b_{\delta n}}f_{\cC_n}\|U_n) \leq D_{\infty}(T_{b_{\delta n}}f_{\cC_n}\|U_n) = o(1).
\end{align*}
Hence the equality case in \eqref{eq: C<R smoothing} is achievable for all $\alpha \in [0,\infty]$. 

Let us prove \eqref{eq: C>R smoothing}. From Lemma \ref{lemma: D_q(T_r|U) impos}, we have 
\begin{align*}
    D_{\alpha}(T_{b_{\delta n}}f_{\cC_n}\|U_n) &\geq n(1-R_n)-H_{\alpha}(b_{\delta n}) \geq n(1-R_n-h(\delta))
\end{align*}
because $\frac{1}{n}H_{\alpha}(b_{\delta n}) = \frac{1}{n}\log V_{\delta n} \leq  h(\delta)$. 
  
We are left to show that for $R<1-h(\delta),$ \eqref{eq: C>R smoothing} can be achieved with equality in the limit of large $n$. We use a random coding argument to prove this. Let $\cC_n$ be an $(n,2^{nR_n})$ code whose codewords are chosen independently and uniformly. In Eq.~\eqref{eq:Q}, Appendix B, we define the expected norm of the noisy function. Here we use this quantity for the ball noise kernel. For $\alpha \in [0,\infty)$, define
$$
Q_n(\alpha) = \mathbbm{E}_{\cC_n}2^{(\alpha-1)D_{\alpha}(T_{b_{\delta n}}\|U_n)}.
$$ 
From Lemma \ref{lemma: recursive}, for any rational $\alpha \geq 1$, 
\begin{align}\label{eq: rec_ball}
    Q_n(\alpha) & \leq \sum_{k=0}^p \binom{p}{k} 2^{\frac{nk}{q}(1-R_n-\frac{\log V_{\delta n}}{n})}Q_n\Big(\frac{p-k}{q}\Big),
\end{align}
for $p,q \in \mathbbm{Z}_{0}^{+}$ such that $\alpha = 1+\frac pq$.

 Assume that $R <1-h(\delta)$.  Let us prove that $Q_n(\alpha) \leq  2^{n(\alpha-1)(1-R-h(\delta) +o(1))}$ for rational values of $\alpha$ using induction. Let $\alpha \in [1,2]$ be rational and note that $p\leq q$. Since $Q_n(\cdot) \leq 1$ when the argument is less than 1, we can write \eqref{eq: rec_ball} as
follows:
\begin{align*}
    Q_n(\alpha) & \leq \sum_{k=0}^p \binom{p}{k} 2^{\frac{nk}{q}(1-R_n-\frac{\log V_{\delta n}}{n})} = 2^{n(\alpha-1)(1-R-h(\delta) +o(1))}.
\end{align*}
Now assume that \eqref{eq: rec_ball} holds for all rational $\alpha \in [1,m]$ for some integer  $m \geq 2$ and
prove that in this case it holds also for $\alpha\in(m,m+1].$ 
By the induction hypothesis,
 \begin{align*}
    Q_n(\alpha) & \leq \sum_{0\leq k\leq p-q} \binom{p}{k} 2^{\frac{nk}{q}(1-R_n-\frac{\log V_{\delta n}}{n})}2^{n\frac{p-k-q}{q}(1-R-h(\delta) +o(1))} +\sum_{k = p-q}^p \binom{p}{k} 2^{\frac{nk}{q}(1-R_n-\frac{\log V_{\delta n}}{n})} \\
    & \leq \sum_{0\leq k\leq p-q} \binom{p}{k} 2^{n(\alpha-2)(1-R-h(\delta) +o(1))} + \sum_{k =p-q}^p \binom{p}{k} 2^{n(\alpha-1)(1-R-h(\delta)+o(1))}\\
    & = 2^{n(\alpha-1)(1-R-h(\delta) +o(1))}.
\end{align*}
Therefore,  for every rational $\alpha \in (1,\infty)$ there exists a sequence of codes satisfying  
\begin{align}\label{eq: C>R ach}
  D_{\alpha}(T_{b_{\delta n}}f_{\cC_n}\|U_n) = n(1-R-h(\delta) +o(1)),  
\end{align}
which is equivalent to the equality in \eqref{eq: C>R smoothing}. 

Let us extend this result to non-negative reals. Let $\alpha \in [0,\infty)$ and let us choose a rational $\alpha' \in (1,\infty)$ such that $\alpha < \alpha'$.
We know that there exists a sequence of codes satisfying  
\begin{align*}
  D_{\alpha'}(T_{b_{\delta n}}f_{\cC_n}\|U_n) = n(1-R-h(\delta) +o(1)).
\end{align*}
From \eqref{eq: C>R smoothing} and from Remark \ref{rm: D_alpha mono},
\begin{align*}
      n(1-R_n-h(\delta)) \leq  D_{\alpha}(T_{b_{\delta n}}f_{\cC_n}\|U_n) \leq D_{\alpha'}(T_{b_{\delta n}}f_{\cC_n}\|U_n) = n(1-R-h(\delta) +o(1)).
\end{align*}
Hence, the asymptotic equality in \eqref{eq: C>R smoothing} is achievable for all $\alpha \in [0,\infty)$.

\end{proof}

The above theorem characterizes the $D_{\alpha}$-smoothing capacities with respect to ball noise.

\begin{corollary}\label{Cor: ball noise D_alpha cap}
Let $\delta \in [0,1/2]$. Let ${b}(\delta) = \left(b_{\delta n}\right)_n$ be a sequence of ball noise operators, where $\delta n$ is the radius corresponding to the $n$-th kernel. Then
    \begin{align*}
        S_{\alpha}^{{b}(\delta)} &= 1-h(\delta) \quad \text{for } \alpha \in [0,\infty].
    \end{align*}
\end{corollary}

Norms of $T_{b_{t}}f_{\cC}$ can be used to bound decoding error probability on a BSC. 
While estimating these norms for a given code is generally complicated,
the second norm affords a compact expression based on the distance distribution of the code. In the next section, we bound decoding error probability using the second norm of $T_{b_{t}}f_{\cC}$. The following proposition provides closed-form expressions for $\|2^nT_{b_{t}}f_{\cC}\|_2^2$.

\begin{proposition}\label{prop: ball noise L2} 
\begin{align*}
        \|2^nT_{b_{t}}f_{\cC}\|_2^2=\frac{2^n}{|\cC|V_t^2}\sum_{i=0}^n \mu_t(i)A_i = \frac{1}{V_t^2}\sum_{k=0}^n L_t(k)^2A_k^{\bot}.
\end{align*}
where $\mu_t(i)$ is defined in \eqref{eq:mu_t} and $L_t$ is the Lloyd polynomial of degree $t$ \eqref{eq: ft ball}.
\end{proposition}
The proof is immediate from Proposition \ref{prop: L2-smth} in combination with  \eqref{eq:mu_tF} and \eqref{eq:r-hat r}.

\subsection{Probability of decoding error on a BSC\texorpdfstring{$(\delta)$}{} }

The idea that smoothing of codes under some conditions implies good decoding performance has appeared in a number of papers using different language. Smoothing of capacity achieving codes was considered in \cite{shamai1997empirical,polyanskiy2013empirical}. H{\k a}z\l{}a et al. \cite{hkazla2021codes} showed that if a code (sequence) is perfectly smoothable with respect to Bernoulli noise, then the dual code is good for decoding (see Theorem~\ref{prop: Good erasure correcting codes in in BSC}, Corollary~\ref{cor: BEC capacity achieving codes in BSC}).
Going from smoothing to decodability involves representing the $D_2$-smoothness of codes with respect to Bernoulli noise as a potential energy form and comparing it to the Bhattacharya bound for the dual codes. One limitation of this approach is that it cannot infer decodability for rates $R > 1-\log(1+2\sqrt{\delta(1-\delta)})$ (this is the region above the blue solid curve in Figure \ref{fig: Wiretap channel - Thresholds}). Rao and Sprumont \cite{rao2022criterion} and H{\k a}z\l{}a \cite{hkazla2022optimal} proved that sufficient smoothing of codes implies the decodability of the codes themselves rather than their duals. However, these results are concerned with list decoding for rates above the Shannon capacity, resulting in exponential list size, which is arguably less relevant from the perspective of communication. 

Except for \cite{rao2022criterion}, the cited papers utilize  perfect or near-perfect smoothing to infer decodability.  For codes whose rates are below the capacity, perfect smoothing is impossible. At the same time, codes that possess sufficiently good smoothing properties are good for decoding. This property is at the root of the results for list decoding in  \cite{rao2022criterion}; however, their bounds were insufficient to make conclusions about list decoding below capacity.

Consider a channel where, for the input $X \sim f_{\cC}$, the output $Y$ is given by $Y = X + W$ with $W \sim b_t$. Define $F_t(y) = |\cC \cap B(y,t)|$ be the number of codewords in the ball $B(y,t)$. Hence, for a received vector $y$, the possible number of codewords that can yield $y$ is given by $F_t(y)$. Intuitively, the decoding error is small if $F_t(y) \approx 1$ for typical errors. Therefore, $F_t$ is of paramount interest in decoding problems. Since the typical errors for both ball noise and the Bernoulli noise is almost the same, this allows us to obtain a bound for decodability in the BSC channel.
Using this approach, we show that the error probability of decoding
on a ${\text{\rm BSC}}(\delta)$ can be expressed via the second moment of the number of codewords in 
the ball of radius $ t \gtrsim \delta n $.

Assume without loss of generality that $\cC$ is a linear code and $0^n$ is used for transmission.
Let $Y$ be the random Bernoulli vector of errors, and note that $Y \sim \beta_{\delta}$.
The calculation below does not depend on whether we rely on unique or list decoding within a ball of radius $t$, so let us assume that the decoder outputs $L\ge 1$  candidate codewords conditioned on the received vector $y$, which is a realization of $Y.$ 

In this case, the list decoding error can be written as 
\begin{align}\label{eq: lst_error def}
    P_{L,t}(\cC,{\text{\rm BSC}}(\delta)) = \Pr\{F_t(Y) \geq L+1 \cup |Y| > t\}.
\end{align}

\begin{theorem}\label{lemma: list decoding error} 
Let $t$ and $t^{\prime}$ be integers such that $0<t^{\prime} < t<n$. Then for any $L\ge 1,$
\begin{align} \label{eq: ld_error}
    P_{L,t}(\cC,{\text{\rm BSC}}(\delta)) \leq \frac{\beta_{\delta}(t^{\prime})}{L}\sum_{w=1}^{n}{\mu_t(w)}A_w + \Pr(|Y|\leq t^\prime \cup |Y|\geq t ).
\end{align}
\end{theorem}

\begin{proof}
     Define $S_{t^{\prime},t}= B(0,t)\setminus B(0,t^{\prime})$. Clearly,
\begin{align*}
    P_{L,t}(\cC,{\text{\rm BSC}}(\delta)) &= \Pr\{F_t(Y) \geq L+1 \cup |Y| > t\}\\
 & \le \Pr\{(F_t(Y) \geq L+1) \cap (Y\in S_{t^{\prime},t})\} + \Pr(Y\not\in S_{t^{\prime},t}).
   \end{align*}
 Let us estimate the first of these probabilities. 
  \begin{align*}
   & \Pr\{(F_t(Y) \geq L+1) \cap (Y\in S_{t^{\prime},t})\}\\ 
   & = \sum_{y\in S_{t^{\prime},t}}\1_{F_t(y)\geq L+1}\beta_{\delta}(y)\\  
   & \leq \sum_{y\in S_{t^{\prime},t}}\frac{F_t(y)-1}{L}\beta_{\delta}(y)
   \\
    & \leq \frac{\beta_{\delta}(t^{\prime})}{L}\sum_{y\in S_{t^{\prime},t}}(F_t(y)-1)    \\
    & \leq \frac{\beta_{\delta}(t^{\prime})}{L}\sum_{y\in B(0,t)}(F_t(y)-1) \quad \text{(because for all  } y \in B(0,t),  F_t(y) \geq 1)  \\ 
    & = \frac{\beta_{\delta}(t^{\prime})}{L}\Big(\sum_{y \in \cH _n}(\1_{\cC}\ast\1_{B(0,t)})(y)\1_{B(0,t)}(y) -V_t\Big)\\
    & = \frac{\beta_{\delta}(t^{\prime})}{L}\Big(\sum_{c \in \cC}(\1_{B(0,t)}\ast\1_{B(0,t)})(c) -V_t\Big)\\
    & = \frac{\beta_{\delta}(t^{\prime})}{L}\sum_{i =1}^n\mu_t(i)A_i. \qedhere
\end{align*}
\end{proof}

\begin{remark}
In the case of $L=1$, the bound in \eqref{eq: ld_error} can be considered a slightly weaker version of Poltyrev's bound \cite{poltyrev1994bounds}, Lemma 1. By allowing this weakening, we obtain a bound in
a somewhat more closed-form, also connecting decodability with smoothing. We also prove a simple bound for the error probability of list decoding expressed in terms of the code's distance distribution (and, from \eqref{eq:r-hat r}, also in terms of the dual distance distribution). The latter result seems not to have appeared in earlier literature.
\end{remark}

The following version of this lemma provides an error bound, which is useful in the asymptotic setting.  
\begin{proposition}  \label{prop: list decoding error - asymptotic}
Let $t = \delta n+ n^{\theta},$ where $\theta\in(1/2,1).$ then
\begin{align*} 
    P_{L,t}(\cC,{\text{\rm BSC}}(\delta)) \leq \frac{\sqrt{2n}}{L{V_t}}\Big(\frac{1-\delta}{\delta}\Big)^{2n^{\theta}}\sum_{w=1}^{n}{\mu_t(w)}A_w +2e^{-n^{2\theta-1}}.
\end{align*}

In particular,
\begin{align*} 
    P_{L,t}(\cC,{\text{\rm BSC}}(\delta)) \leq \frac{\sqrt{2n}}{V_t}\Big(\frac{1-\delta}{\delta}\Big)^{2n^{\theta}}\sum_{w=1}^{n}{\mu_t(w)}A_w +2e^{-n^{2\theta-1}}.
\end{align*}
\end{proposition} 

\begin{proof}
Set $t^{\prime} = \delta n -n^{\theta}$. A direct calculation shows that $$\beta_\delta(t^{\prime})V_t < {\sqrt{2n}}\Big(\frac{1-\delta}{\delta}\Big)^{2n^{\theta}}.$$
By the Hoeffding bound, 
\begin{align*}
   \Pr(|Y|\leq t^\prime \cup |Y|\geq t ) \le 2e^{-n^{2\theta-1}}.
\end{align*}
Together with Lemma \ref{lemma: list decoding error} this implies our statements.
\end{proof}

A question of prime importance is whether the right-hand side quantities in Proposition \ref{prop: list decoding error - asymptotic} converge to 0. For $R < 1-h(\delta)$, one can easily see that for random codes, $\sum_{w=1}^{n}\frac{\mu_t(w)}{V_t}A_w = 2^{-\Theta(n)}$ where $t = \delta n+ n^{\theta}$ showing that this is in fact the case. 

From Proposition \ref{prop: ball noise L2}, it is clear that the potential energy $\sum_{w=1}^{n}{\mu_t(w)}A_w$ is a measure of the smoothness of $T_{b_t}f_\cC$. This implies that codes that are sufficiently smoothable with respect to $b_t$ are decodable in the BSC with vanishing error probability. In other words, Proposition \ref{prop: list decoding error - asymptotic} establishes a connection between smoothing and decoding error probability.


\section{Perfect smoothing---The finite case}\label{sec: Perfect smoothing -finite}
In this section, we briefly overview another form of perfect smoothing, which is historically the
earliest application of these ideas in coding theory. It is not immediately related to the information-theoretic problems considered in the other parts. 

We are interested in radial kernels that yield perfect smoothing for a given code. 
We call $\rho(r):=\max(i: r(i)\ne 0)$ the {\em radius} of $r$. Note that the logarithm of the support size of $r$ (as a function on the space $\cH_n$) is exactly the $0$-R{\'e}nyi entropy of $r$. Therefore, kernels with smaller radii can be perceived as less random, supporting the view of the radius $\rho(r)$ as a general measure of randomness.
\begin{definition}
  We say a code $\cC$ is perfectly smoothable with respect to $r$ if $T_rf_{\cC}(x) = \frac{1}{2^n}$ for all
  $x\in \cH_n$, and in this case we say that $r$ is a perfectly smoothing kernel for $\cC$.
\end{definition}

Intuitively, such a kernel should have a sufficiently large radius. In particular, it should be as large as the {\em covering radius} of the code $\rho(\cC)$ or otherwise, smoothing does not affect the vectors that are $\rho$ away from the code. To
obtain a stronger condition, recall that the external distance of code $\cC$ is 
    $
         \bar{d}(\cC) = |\{i\ge 1: A_i^\bot\ne0\}|.
    $
\begin{proposition}\label{prop: perfect smoothing requires spread kernels}
    Let $r$ be a perfectly smoothing kernel of code $\cC$. Then $ \rho(r) \geq \bar{d}(\cC).$
\end{proposition}

\begin{proof} Note that perfect smoothing of $\cC$ with respect to $r$ is equivalent to 
\begin{align*}
    \|2^nT_rf_{\cC}\|_2^2 = 1,
\end{align*}
which  by Proposition \ref{prop: L2-smth} is equivalent to the following condition:
\begin{align*}
    \sum_{i=1}^n\hat{r}(i)^2A^\bot_i = 0.
\end{align*}
Therefore,
\begin{align*}
    \bar{d}(\cC) = |\{i\ge 1: A_i^\bot\ne0\}| \leq n - |\{i\ge 1: \hat{r}(i)\ne0\}|.
\end{align*}
 By definition,
\begin{align*}
    \hat{r} = \frac{1}{2^n} K^{\intercal}r,
\end{align*}
where $K=(K_i(j))_{i,j=0}^n$ is the Krawtchouk matrix. 
Define $I_1 = \{j\in \{1,2,\dots, n\} : \hat{r}(j) = 0 \}$ and $I_2 = \{i\in \{1,2,\dots, n\} : r(i) \neq 0 \} $
then 
\begin{align*}
    0 = \hat{r}|_{I_1} = \frac{1}{2^n}K^{\intercal}|_{(I_1,:)}r = \frac{1}{2^n} K^{\intercal}|_{(I_1,I_2)}r|_{I_2}.
\end{align*}
This relation implies that there exists a linear combination of Krawtchouk polynomials of degree at most $\rho(r)$ with $|I_1|$ roots. Therefore, $\bar{d}(\cC) \leq
n - |\supp(\{\hat{r}(i)\}_{i=1}^n)| = |I_1| \leq \rho(r).$ 
\end{proof}
Since $\rho(\cC)\le \bar d(\cC)$, this inequality strengthens the obvious condition $\rho(r) \geq \rho(\cC).$ 
At the same time, there are codes that are perfectly smoothable by a radial kernel $r$ such that $\rho(r)=\rho(\cC)$. 

\begin{definition}{\cite{semakov1971uniformly}} \label{def:UP}
A code $\cC$ is uniformly packed in the wide sense if there exists rational numbers $\{\alpha_i\}_{i=0}^{\rho}$ such that 
\begin{align*}
    \sum_{i=0}^{\rho(\cC)} \alpha_iA_i(x) = 1 \quad \text{ for all } x\in \cH_n,
\end{align*}
where $A_i(x)$ is the weight distribution of the code $\cC-x$.
\end{definition}
Our main observation here is that some uniformly packed codes are perfectly smoothable with respect to noise kernels that are minimal in a sense. The following proposition states this more precisely. 
\begin{proposition}\label{prop: uniformly packed codes} Let $\cC$ be a code that is perfectly smoothable 
by a radial kernel of radius $\rho(r)=\rho(\cC)$. Then $\cC$ is uniformly packed in the wide sense with
$\alpha_i\ge 0$ for all $i$.
\end{proposition}
\begin{proof}
  By definition, if $\cC$ is perfectly smoothable with respect to $r$, then $2^nT_rf_{\cC}=1$ which is tantamount to $\sum_{y \in \cH_n}\frac{2^n}{|\cC|}r(y)\1_{\cC}(x-y)= 1$ for all $x \in \cH_n$. This condition can be written as $\sum_{i=0}^{\rho}\big(\frac{2^n}{|\cC|}r(i)\big)A_i(x)= 1$ for all $x \in \cH_n$, completing the proof.
\end{proof}

To illustrate this claim, we list several families of uniformly packed codes (\cite{semakov1971uniformly,goethals1975uniformily}; see also \cite{tokareva2007upper}) that are perfectly 
smoothable by a kernel of radius equal to the covering radius of the code.

    \begin{enumerate}[label = (\roman*)]
        \item Perfect codes: $r = b_{\rho}$ where $\rho=\rho(\cC)$ is the covering radius. 
        \item 2-error correcting BCH codes of length $2^{2m+1}, m\ge 2$. The smoothing kernel $r$ is given by
        $$r(0)=r(1) = L, r(2)=r(3)=\frac{3L}{n}, r(i)=0, i\geq 4.$$
        \item Preparata codes. The smoothing kernel $r$ is given by
        $$r(0)=r(1) = L, r(2)=r(3)=\frac{6L}{n-1}, r(i)=0, i\geq 4.$$
        \item Binary $(2^m-1,2^{2^m-3m+2},7)$ Goethals-like codes \cite{goethals1975uniformily}. The smoothing kernel $r$ is given by
        $$r(0)=r(1) = L, r(2)=r(3)=\frac{65L}{2n}, r(4)=r(5)=\frac{30L}{n(n-3)}, r(i)=0, i\geq 4.$$
  \end{enumerate}
Here $L$ is a generic notation for the normalizing factor.
More examples are found in a related class of {\em completely regular codes} \cite{borges2019completely}.

Definition~\ref{def:UP} does not include the condition that $\alpha_i\ge0$, and in fact there are codes that are uniformly packed in the wide sense, but some of the
$\alpha_i$'s are negative, and thus they are not smoothable by a noise kernel of radius
$\rho(\cC)$. One such family is 3-error-correcting binary BCH codes of length $2^{2m+1}, m\ge 2$ \cite{goethals1975uniformily}.

\appendix

\section{$L_2$ smoothing}

The Fourier transform of a function $f:\cH _n \rightarrow \reals$ is a function on the dual group $\widehat\cH_n$, which we identify with $\cH_n$:
\begin{equation}\label{eq:ft}
 \widehat{f}(y) = \frac{1}{2^n}\sum_{x\in {\cH _n}} f(x)(-1)^{x \cdot y}, \quad y\in\cH_n.
\end{equation}

The Fourier transform of the indicator function of the sphere is given by $\widehat\1_{S(0,t)}=\frac1{2^n}K_t,$ where
$K_t(x)= {K}_t^{(n)}(x) = \sum_{j=0}^t(-1)^j\binom{x}{j}\binom{n-x}{t-j}$ is a Krawtchouk polynomial of degree $t$. Then clearly
the Fourier transform of the indicator of the ball is
\begin{equation}\label{eq: ft ball}
    \widehat\1_{B(0,t)}  =  \frac{1}{2^n}L_t,
\end{equation}
where $L_t(x):=\sum_{i=0}^t K_i(x)$ is called the Lloyd polynomial \cite[p.64]{delsarte1973an}. The intersection of balls in \eqref{eq:mu_t} can be written as $\1_{B(0,t)}\ast\1_{B(x,t)},$ which implies the expression \cite[Lemma~4.1]{barg2021stolarsky}
   \begin{equation}\label{eq:mu_tF}
   \mu_t(i)=2^{-n}\sum_{k=0}^n L_t(k)^2 K_k(i).
   \end{equation}
Given a code $\cC\subset \cH_n$, we define the {\em dual distance distribution} of $\cC$ 
as the set of numbers $A^\bot_j:=\frac{1}{|C|}\sum_{i=0}^n A_i K_j(i),$ where $(A_i)_{i=0}^n$ is the distance distribution of $\cC$ \eqref{eq:dd}. Note that when $\cC$ is linear, the set
$(A^\bot_j)_{j=0}^n$ coincides with the distance distribution of its dual code $\cC^\bot.$
For a radial potential $V$ on $\cH_n$ and a code $\cC$ we have
    \begin{align}\label{eq:r-hat r}
        \sum_{i=0}^nV(i)A_i = |\cC|\sum_{k=0}^n\widehat{V}(k)A_k^{\bot}.
    \end{align}

The $ L_2$-smoothness of a noisy code distribution can be written in terms of the distance distribution or of the dual distance distribution.
\begin{proposition}\label{prop: L2-smth}
Let $\cC$ be a code and $r$ be a noise kernel. Then
\begin{align*}
    \|2^nT_rf_{\cC}\|_2^2 = \frac{2^n}{|\cC|}\sum_{i=0}^n (r \ast r)(i)A_i = 4^n\sum_{k=0}^n\hat{r}(k)^2A_k^\bot.
\end{align*}

\end{proposition}

\begin{proof}
Let us prove the first equality:
    \begin{align}
    \|2^nT_rf_{\cC}\|_2^2 &= \frac{1}{2^n}\sum_{x\in \cH_n}(2^nT_rf_{\cC}(x))^2\nonumber\\
    &= \frac{2^n}{|\cC|^2V_t^2}\sum_{x\in \cH_n}(r\ast \1_{\cC})(x)^2\nonumber\\
    &= \frac{2^n}{|\cC|^2}\sum_{x\in \cH_n}\sum_{y\in \cH_n}r(x-y) \1_{\cC}(y)\sum_{z\in \cH_n}r(x-z)\1_{\cC}(z)\nonumber\\
    &= \frac{2^n}{|\cC|^2}\sum_{y\in \cC} \sum_{z\in \cC}\sum_{x\in \cH_n}r(x-y)r(x-z)\nonumber\\
    &= \frac{2^n}{|\cC|^2}\sum_{y\in \cC} \sum_{z\in \cC}(r\ast r)(y-z)\nonumber\\
    &= \frac{2^n}{|\cC|}\sum_{i=0}^n (r \ast r)(i)A_i.  
\end{align}
The second equality is immediate by noticing that $\widehat{r\ast r} = 2^n \hat{r}^2$ and using \eqref{eq:r-hat r}.
\end{proof}

\vspace*{.1in}\section{Proof of Theorem ~\ref{thm: General smoothing capacities}}\label{appendixB}

We will first establish Theorem~\ref{thm: General smoothing capacities} when  $\alpha$ is rational, and then use a density argument to extend the proof to all real numbers. The case $\alpha= \infty$ is handled separately at the end of this appendix.

We will use the following technical claim.

\begin{lemma}\label{lemma: frac binom.}
    Let $x$ and $y$ be two non-negative reals. Further, let $p$ and $q$ be positive integers. Then 
    \begin{align*}
        (x+y)^\frac{p}{q} \leq \sum_{k=0}^p\binom{p}{k}x^\frac{k}{q}y^\frac{p-k}{q}.
    \end{align*}
\end{lemma}
\begin{proof}
    Clearly $(x+y)^\frac{1}{q} \leq x^\frac{1}{q}+y^\frac{1}{q}$. Therefore,
    $$
        (x+y)^\frac{p}{q} \leq ( x^\frac{1}{q}+y^\frac{1}{q})^p= \sum_{k=0}^p\binom{p}{k}x^\frac{k}{q}y^\frac{p-k}{q}.\qedhere
    $$
\end{proof}
For $M\ge 1$ let $\cC = (c_0,c_2,\dots,c_{M-1})$ be a code whose codewords are chosen randomly and independently from $\cH_n.$ For $\alpha \in [0,\infty)$, define 
   \begin{equation}\label{eq:Q}
   Q_n(\alpha) = 
   \mathbbm{E}_{\cC}2^{(\alpha-1)D_{\alpha}(T_rf_\cC\|U_n)}.
   \end{equation} 
For $\alpha>0,$ $Q_n(\alpha)=\|2^nT_rf_\cC(x)\|^\alpha_\alpha.$ Clearly $Q_n(1) = 1, Q_n(\alpha) \leq 1 \text{ for } \alpha \in [0,1)$, and $Q_n(\alpha) \geq 1$
for $\alpha>1.$ 

In the next lemma, we obtain a recursive bound for $Q_n$. We will then use an induction argument to show the full result.

\begin{lemma}\label{lemma: recursive}
Let $\alpha = \frac{p}{q}+1$ and let $\cC\subset\cH_n$ be a random code of size $M=2^{nR}$. Then
\begin{align}
    Q_n(\alpha) & \leq \sum_{k=0}^p \binom{p}{k} 2^{\frac{nk}{q}(1-R-\frac1nH_{1+k/q}
    (r))}Q_n\Big(\frac{p-k}{q}\Big). \label{eq:Qn}
\end{align}
\end{lemma} 

\begin{proof} 
In the calculation below we write $\mathbbm{E}$ for $\mathbbm{E}_\cC$. Starting with \eqref{eq:Q}, we obtain
\begin{align*}
     Q_n(\alpha) 
     & = \mathbbm{E}\Big[\frac{1}{2^n}\sum_{x \in \cH_n} [2^n(r \ast f_\cC)(x)]^{\alpha}\Big]\\
     & = 2^{n(\alpha-1)}\mathbbm{E}\Big[\sum_{x \in \cH_n} \Big[\sum_{z \in \cC}r(x-z)\frac{1}{M}\Big]^{\alpha}\Big]\\
     & = \frac{2^{n(\alpha-1)}}{M^{\alpha}}\sum_{x \in \cH_n} \mathbbm{E}\Big[\sum_{i = 0}^{M-1}r(x-c_i)\Big]^{\alpha}\\
     & = \frac{2^{n(\alpha-1)}}{M^{\alpha}}\sum_{x \in \cH_n} \mathbbm{E}\Big[\sum_{i = 0}^{M-1}r(x-c_i)\Big[\sum_{j = 0}^{M-1}r(x-c_j)\Big]^{\alpha-1}\Big]\\
     & = \frac{2^{n(\alpha-1)}}{M^{\alpha}}\sum_{x \in \cH_n} \mathbbm{E}\Big[\sum_{i = 0}^{M-1}r(x-c_i)\Big[r(x-c_i) + \sum_{j = 0, j\neq i}^{M-1}r(x-c_j)\Big]^{\frac{p}{q}}\Big]\\
     & \leq \frac{2^{n(\alpha-1)}}{M^{\alpha}}\sum_{x \in \cH_n} \mathbbm{E}\Big[\sum_{i = 0}^{M-1}r(x-c_i)\sum_{k=0}^p \binom{p}{k} r(x-c_i)^\frac{k}{q} \Big[\sum_{j = 0, j\neq i}^{M-1}r(x-c_j)\Big]^\frac{p-k}{q}\Big]\\
     & = \frac{2^{n(\alpha-1)}}{M^{\alpha}}\sum_{k=0}^p \binom{p}{k}\sum_{x \in \cH_n} \mathbbm{E}\Big[ \sum_{i = 0}^{M-1}r(x-c_i)^{1+\frac{k}{q}} \Big[\sum_{j = 0, j\neq i}^{M-1}r(x-c_j)\Big]^\frac{p-k}{q}\Big]\\
     & = \frac{2^{n(\alpha-1)}}{M^{\alpha}}\sum_{k=0}^p \binom{p}{k}\sum_{x \in \cH_n} \mathbbm{E}\Big[ \sum_{i = 0}^{M-1}r(x-c_i)^{1+\frac{k}{q}}\Big] \;\mathbbm{E}\Big[\sum_{j = 0, j\neq i}^{M-1}r(x-c_j)\Big]^\frac{p-k}{q},
\end{align*}
where $c_i,i=1,\dots,M$ are random codewords in the code $\cC.$ Recalling that $\mathbbm{E} r(x-c_i)^{a} =  \|r\|_a^a$ for any $a>0,$ we continue as follows:
\begin{align*} 
     & \leq \frac{2^{n(\alpha-1)}}{M^{\alpha}}\sum_{k=0}^p \binom{p}{k}\sum_{x \in \cH_n} M\|r\|_{1+k/q}^{1+k/q}\; \mathbbm{E}\Big[\sum_{j = 0}^{M-1}r(x-c_j)\Big]^\frac{p-k}{q}\\
     & = \frac{2^{n(\alpha-1)}}{M^{\alpha-1}}\sum_{k=0}^p \binom{p}{k} \|r\|_{1+k/q}^{1+k/q} \; \mathbbm{E}\Big[\sum_{x \in \cH_n}\Big[\sum_{j = 0}^{M-1}r(x-c_j)\Big]^\frac{p-k}{q}\Big]\\
     & = \frac{2^{np/q}}{M^{p/q}}\sum_{k=0}^p \binom{p}{k} \|r\|_{1+k/q}^{1+k/q} Q_n\Big(\frac{p-k}{q}\Big)\frac{M^{(p-k)/q}}{2^{n((p-k)/q -1)}}\\
     & = \sum_{k=0}^p \binom{p}{k} \frac{2^{n(1+k/q)}}{M^{k/q}} \|r\|_{1+k/q}^{1+k/q}Q_n\Big(\frac{p-k}{q}\Big)\\
     & = \sum_{k=0}^p \binom{p}{k} 2^{\frac{nk}{q}\big(1-R-\frac{H_{(1+k/q)}(r)}{n}\big)}Q_n\Big(\frac{p-k}{q}\Big),
\end{align*}
where we used \eqref{eq:Renyi} and the fact that $r$ is a pmf.
\end{proof}

On account of \eqref{eq:smooth-p}, \eqref{eq:Q}, and Lemma~\ref{lemma: D_q(T_r|U) impos}, to prove Theorem~\ref{thm: General smoothing capacities}, we need to prove the following
\begin{theorem}\label{thm: main} Consider a sequence of ensembles of random codes of increasing length $n$ and rate $R_n\to R.$
   If $R > 1-\pi(\alpha)$, where $\pi(\alpha)$ is given by \eqref{eq:pi}, then
   \begin{align}\label{eq: main}
       \lim_{n \to \infty}Q_n(\alpha) = 1
   \end{align}
   for all $\alpha \in (1,\infty)$. 
   \end{theorem}
We start with the case of rational $\alpha$.
\begin{proposition}   
    Let $\alpha\ge 0$ be rational. If $R > 1-\pi(\alpha)$ then $\limsup_nQ_n(\alpha) \leq 1$.
\end{proposition}
\begin{proof}
This statement is true for all $0\le\alpha<1,$ so also for all rational $\alpha$ in [0,1)
Assume that it holds for all rational $\alpha$ in $[0,m)$ where $m \in \mathbbm{Z}^{+}$.
    Let  $\alpha \in [m,m+1)$ and choose $p,q \in \mathbbm{Z}_{0}^{+}$ such that $\alpha = 1 + \frac{p}{q}$.  By Lemma \ref{lemma: recursive},
    \begin{align*}
        \limsup_nQ_n(\alpha) &\leq \limsup_n\sum_{k=0}^p \binom{p}{k} 2^{\frac{nk}{q}\big(1-R_n-\frac{H_{1+k/q}(r_n)}{n}\big)}Q_n\Big(\frac{p-k}{q}\Big)\\
        & \leq \sum_{k=0}^p \binom{p}{k} \limsup_n2^{\frac{nk}{q}\big(1-R_n-\frac{H_{1+k/q}(r_n)}{n}\big)}\limsup_n Q_n\Big(\frac{p-k}{q}\Big).
    \end{align*}
If $R > 1-\pi(\alpha)$, then evidently, $R > 1-\pi(1+k/q)$ for all $k\le p$. Therefore,  
  $$
  \limsup_{n \to \infty} 2^{\frac{nk}{q}\big(1-R_n-\frac{H_{1+k/q}(r_n)}{n}\big)} =0 
  $$
for all $k>0$. Since $\frac{p}{q}<m,$ by the induction hypothesis we have $ \limsup_n Q_n\big(\frac{p-k}
{q}\big) \leq 1$ for $k=0, 1, \dots, p$. Therefore, all the terms except the one with $k=0$ vanish, yielding $\limsup_n Q_n(\alpha) \leq 1$. 
\end{proof}

Since $Q_n(\alpha) \geq 1$ for $\alpha > 1$, this proves Theorem \ref{thm: main}
for all rational $\alpha \in (1, \infty).$

Finally, let us extend this result to all real $\alpha > 1$. As a first step, let us show that $\pi(\alpha)$ is continuous.

\begin{lemma}
    $\pi(\alpha) $ is continuous for $1<\alpha<\infty.$ 
\end{lemma}
\begin{proof}
From the monotonicity of R{\'e}nyi entropies, for $\alpha' > \alpha >1$, 
\begin{align*}
    0 &\leq \pi(\alpha)-\pi(\alpha')\\
    & = \liminf_{n \to \infty}\frac{1}{n}H_{\alpha}(r_n)-\liminf_{n \to \infty}\frac{1}{n}H_{\alpha'}(r_n).
\end{align*}
Now let us choose a subsequence $(r_{n_k})_k$ such that 
\begin{align*}
    \lim_{k\to \infty}\frac{1}{n_k}H_{\alpha'}(r_{n_k}) = \liminf_{n \to \infty}\frac{1}{n}H_{\alpha'}(r_n).
\end{align*}
Therefore, 
\begin{align*}
    \pi(\alpha)-\pi(\alpha') &= \liminf_{n \to \infty}\frac{1}{n}H_{\alpha}(r_n)-\lim_{k \to \infty}\frac{1}{n_k}H_{\alpha'}(r_{n_k})\\
    &\leq \liminf_{k \to \infty}\frac{1}{n_k}(H_{\alpha}(r_{n_k})-H_{\alpha'}(r_{n_k})).
\end{align*}
Note that $H_\alpha$ is a continuous function of the order $\alpha$ for $\alpha>1$. 
We use the mean value theorem to claim that there is a value $\gamma_k\in(\alpha,\alpha'$ such that
$H_{\alpha'}(r_{n_k})-H_{\alpha}(r_{n_k})=(\alpha'-\alpha)\frac{d}{d\alpha}H_{\gamma_k}(r_{n_k}).$
Next, for any probability vector $P$, 
   $$-\frac{dH_{\alpha}(P)}{d\alpha} = \frac{1}{(1-\alpha)^2}D(Z\|P) \leq  \frac{\log|\text{supp}(P)|}{(1-\alpha)^2} , 
   $$
where $Z_i = \frac{P_i^{\alpha}}{\sum_{j}P_j^{\alpha}}.$ Taking these remarks together, we obtain
\begin{align*}
    \pi(\alpha)-\pi(\alpha')&\le \liminf_{k \to \infty}\frac{1}{n_k}(\alpha - \alpha')H_{\gamma_k}'(r_{n_k}) \\
    &\leq \liminf_{k \to \infty}\frac{1}{n_k}(\alpha' - \alpha)\frac{n_k}{(\gamma_k-1)^2} \\
    &= \frac{\alpha' - \alpha}{(\alpha-1)^2},
\end{align*}
Therefore, $\pi(\alpha)$ is continuous on $(1,\infty)$.
\end{proof}
   
Now let $\alpha\in (1,\infty)$ and assume $R > 1-\pi(\alpha)$. Choose $\alpha' > \alpha$ such that $\alpha'$ is rational and $R > 1-\pi(\alpha')$. This is possible from the continuity of $\pi$. Therefore, 
\begin{align*}
    1 \leq \limsup_{n} Q_n(\alpha) \leq\limsup_{n}  Q_n(\alpha') =1,
\end{align*}
which proves that \eqref{eq: main} and Theorem~\ref{thm: General smoothing capacities} hold for all $\alpha \in [1,\infty)$.

It remains to address the case $\alpha =\infty$. We obtain the following upper bound, whose proof is inspired by \cite{yu2018renyi}.

\begin{lemma}\label{lemma: D_inf(T_r|U) ach.}
Let $\epsilon >0$. We have
    \begin{align}\label{eq:inf}
        \mathbbm{E}_{\cC} \|2^nT_rf_\cC\|_\infty \leq 1+\epsilon+ 2^{2n-H_{\infty}(r)}e^{-2\epsilon^2 2^{-[n(1-R)-H_{\infty}(r)]}}.
    \end{align}
\end{lemma}

\begin{proof} 
Let $\epsilon >0$, then
\begin{align}
   \mathbbm{E}_{\cC} \|2^nT_rf_\cC\|_{\infty}
    &= \mathbbm{E}_{\cC} \left[\|2^nT_rf_\cC\|_{\infty}\mathbbm{1}_{\{\|2^nT_rf_\cC\|_{\infty} \geq 1+ \epsilon\}}\right]\notag \\
    & \quad \quad + \mathbbm{E}_{\cC} \left[\|2^nT_rf_\cC\|_{\infty}\mathbbm{1}_{\{\|2^nT_rf_\cC\|_{\infty} < 1+ \epsilon\}}\right]\notag \\
    & \leq \mathbbm{E}_{\cC} \left[\|2^nT_rf_\cC\|_{\infty}\mathbbm{1}_{\{\|2^nT_rf_\cC\|_{\infty} \geq 1+ \epsilon\}}\right]+ 1+\epsilon\notag \\
    & \leq \mathbbm{E}_{\cC} \left[\|2^nr\|_{\infty}\mathbbm{1}_{\{\|2^nT_rf_\cC\|_{\infty} \geq 1+ \epsilon\}}\right]+ 1+\epsilon\notag \\
    & = \|2^nr\|_{\infty}\Pr_{\cC} \Big\{\max_{y \in \cH_n}2^nT_rf_\cC(y) \geq 1+ \epsilon\Big\}+ 1+\epsilon\notag \\
    & \leq\|2^nr\|_{\infty}2^n\max_{y \in \cH_n}\Pr_{\cC} \Big\{2^nT_rf_\cC(y) \geq 1+ \epsilon\Big\}+ 1+\epsilon. \label{eq:EC}
\end{align}

For any $y \in \cH_n$,
\begin{align*}
    \Pr_{\cC} \{2^nT_rf_\cC(y) \geq 1+ \epsilon\} &= \Pr_{\cC} \Big\{\frac{2^n}{M}\sum_{z \in \cC}r(y-z) \geq 1+ \epsilon\Big\}  \\
    &= \Pr_{c_i \sim U_n, \text{ iid }} \Big\{\frac{2^n}{M}\sum_{i =1}^Mr(y-c_i) \geq 1+ \epsilon\Big\} \\
    &= \Pr \Big\{\frac{2^n}{M}\sum_{i =1}^Mr(c_i) \geq 1+ \epsilon\Big\}\\
    &= \Pr \Big\{\sum_{i =1}^M(2^nr(c_i)-1) \geq M{\epsilon}\Big\} \\
    & \leq \exp\Big\{-\frac{2M^2\epsilon^2}{ M\|2^nr\|_{\infty}}\Big\}
\end{align*}
by the Chernoff bound. Now substitute this estimate into \eqref{eq:EC} and
simplify using \eqref{eq: Renyi div}, \eqref{eq: Renyi smoothness-inf}, and the equality $D_\infty(r\|U_n)=n-H_\infty(r).$ This yields the claimed estimate \eqref{eq:inf}.
\end{proof}
Now let us consider a sequence of (ensembles of) random codes of increasing length $n$ and rate $R_n \to R$. Recalling the definition of $\pi(\cdot)$ in \eqref{eq:pi}, for $n\to\infty$ we obtain 
\begin{align}\label{eq:inf2}
    \limsup_n\mathbbm{E}_{\cC_n} \|2^nT_{r_n}f_{\cC_n}\|_\infty \leq 1+\epsilon
\end{align}
once $R > 1-\pi(\infty).$ Since $\epsilon$ is arbitrarily small, the left-hand side of \eqref{eq:inf2} approaches one, and together with \eqref{eq:smooth-p} this completes the proof of Theorem~\ref{thm: General smoothing capacities}.

\section{Samorodnitsky's inequalities and their implications}\label{Appendix: 
 Sam}
Samorodnitsky \cite{samorodnitsky2016entropy,samorodnitsky2020improved} recently proved certain 
powerful inequalities for $\alpha$-norms of noisy functions, which permit us to estimate proximity to uniformity upon action of Bernoulli noise kernels. 
We state some of them in this appendix after introducing a few more elements of notation. These results are used in Theorem \ref{thm: Wiretap channels achievable rates for BEC capacity-achieving codes} and in Appendix \ref{appendix: proof: lemma: smoothing and erasure}, where we prove Lemma \ref{lemma: smoothing and erasure}.

For a subset $\Gamma\subset [n],$ write $x|_\Gamma$ to denote the coordinate projection of a vector $x\in\cH_n$ on $\Gamma.$ If the subset $\Gamma$ is formed by random choice with $\Pr(i\in\Gamma)=\lambda$ independently for all $i\in[n],$ we write $\Gamma\sim\lambda$. For a function $f$ on 
$\cH_n$ let 
 \begin{equation} \label{eq:E}
   \mathbbm{E}(f|\Gamma)(x) = \frac{1}{2^{n-|\Gamma|}}\sum_{y:y|_\Gamma=x|_\Gamma}f(y).
\end{equation}
Observe that $\mathbbm{E}(f|\Gamma) = f  \ast f_{\cH_{[n]\setminus \Gamma}}, $ where $\cH_S = \{x \in \cH_n: x|_{[n]\setminus S} = 0\}$. Therefore, $\mathbbm{E}(f|\Gamma)(x)$ is the noisy function of $f$ with respect to the pmf given by the indicator function of the subcube $\cH_{[n]\setminus \Gamma}$.  

The {\em entropy} of a function $f:\cH_n\to\reals$ is defined as 
   \begin{equation}\label{eq:Ent-d}
   \Ent[f] = \|f\log f\|_1 - \|f\|_1 \log(\|f\|_1) = \big\|f\log\frac{f}{\|f\|_1}\big\|_1.
   \end{equation}
This quantity can be thought of as the KL divergence between the distribution induced by $f$ on $\cH_n$ and the uniform distribution:
  \begin{equation}\label{eq:E-D}
   \Ent[f] = \|f\|_1D\Big(\frac{f}{\sum f}\|U_n\Big).
  \end{equation}
If $f$ itself is a pmf, then $D(f\|U_n) = 2^n\Ent(f) = \Ent(2^nf).$ 

\begin{theorem}[\cite{samorodnitsky2016entropy}, Corollary~9] \label{thm: Samorodnitsky's inequality for entropy}
    Let $f$ be a nonnegative function on $\cH_n$. then 
    \begin{align*}
        \Ent[T_{\delta}f] \leq \mathbbm{E}_{\Gamma \sim \lambda} \Ent[\mathbbm{E}(f|\Gamma)]. 
    \end{align*}
    where $\lambda  = (1-2\delta)^2.$   
\end{theorem}

\begin{theorem}[\cite{samorodnitsky2020improved}, Theorem~1.1]\label{thm: Samorodnitsky's inequality for norms}
    Let $f$ be a nonnegative function on $\cH_n$ and $\alpha \geq 2$ be an integer. Then
    \begin{equation}\label{eq:tdq}
        \log\|T_{\delta}f\|_{\alpha}  \leq \mathbbm{E}_{\Gamma \sim \lambda}\log\|\mathbbm{E}(f|\Gamma)\|_{\alpha} . 
    \end{equation}
    where $\lambda = \lambda(\alpha,\delta) = 1+\frac{1}{\alpha -1}\log(\delta^{\alpha}  +(1-\delta)^{\alpha} ) = 1-h_{\alpha} (\delta)$.
    Furthermore,
    \begin{equation}\label{eq:tdinf}
        \log\|T_{\delta}f\|_{\infty} \leq \mathbbm{E}_{\Gamma \sim \lambda}\log\|\mathbbm{E}(f|\Gamma)\|_{\infty}. 
    \end{equation}
    where $\lambda = \lambda(\infty,\delta) = 1+\log(1-\delta)= 1-h_{\infty}(\delta) $
\end{theorem}
To interpret inequalities \eqref{eq:tdq}, \eqref{eq:tdinf}, we note that their left-hand side measures the smoothness of the noisy version of $f$ with respect to the noise $\beta_{\delta}$. At the same time, the right-hand side is the average smoothness of the noisy versions of $f$ with respect to the sub-cube pmf's.

H{\k a}z\l{}a \emph{et al.} \cite{hkazla2021codes} used Theorem \ref{thm: Samorodnitsky's inequality for norms} to great effect, showing that if a code corrects erasures up to a certain noise level in a BEC, then with high probability it corrects errors on a BSC channel up to a certain noise level.

\begin{theorem}[\cite{hkazla2021codes}, Corollary~3.4]\label{prop: Good erasure correcting codes in in BSC}
 Let $(\cC_n)_n$ be a sequence of codes whose rate approaches $R$. Assume 
that for some $\lambda \in (0,1-R]$, $P_B({\text{\rm BEC}}(\lambda),\cC_n) = o(\frac{1}{n})$. Then $(\cC_n)_n$ decodes errors on a ${\text{\rm BSC}}(\delta)$ for 
any $\delta$ that satisfies $2\sqrt{\delta(1-\delta)}<2^{\lambda}-1.$\
\end{theorem}
This theorem implies the following corollary.

\begin{corollary}[\cite{hkazla2021codes}]\label{cor: BEC capacity achieving codes in BSC} 
Let $(\cC_n)_n$ be a sequence of codes with rate $R_n\uparrow R$ that recover transmitted messages with high probability  on a ${\text{\rm BEC}}(1-R)$ (i.e., $(\cC_n)_n$ is a capacity achieving sequence for ${\text{\rm BEC}}(1-R)$). Furthermore, assume that $d(\cC_n) = \omega(\log n)$. If $2\sqrt{\delta(1-\delta)}<2^{1-R}-1$, then with high probability, the codes $\cC_n$ correct errors when used on a ${\text{\rm BSC}}(\delta)$ channel. 
\end{corollary}

The authors of \cite{hkazla2021codes} then used this result to show that RM codes of a constant rate correct a nonvanishing proportion of errors on the BSC.

\vspace*{.1in}
\section{Proof of Lemma~\ref{lemma: smoothing and erasure}} \label{appendix: proof: lemma: smoothing and erasure}

We present the proof as a sequence of lemmas.

Let  $\Gamma \subset [n]$ be a subset of coordinates and for $z\in\{0,1\}^n$ let $\cC(\Gamma,z):=\{c\in\cC: c|_\Gamma=z|_\Gamma\}$ be the set of codewords that fit $z$ in the positions of $\Gamma.$ 
In particular, $\cC^{\Gamma^c}=\cC(\Gamma,0)|_{\Gamma^c}$ is the shortened code $\cC$, i.e., the subcode
with zeros in the positions of $\Gamma,$ projected on $\Gamma^c.$
Let $F^{(\cC)}(\Gamma,z):=|\cC(\Gamma,z)|$. 

Let us obtain expressions for the norms and the entropy of $F^{(\cC)}(\Gamma,z).$

\begin{lemma}\label{lemma: Collision function- norms}
 Let $\cC$ be a linear code and let $\Gamma \in [n]$. Then
    \begin{align*} 
        \|F^{(\cC)}(\Gamma,\cdot)\|_\alpha =  \left[  \frac{|\cC|}{2^{|\Gamma|}} F^{(\cC)}(\Gamma,0)^{\alpha -1}\right]^{1/{\alpha} }.
    \end{align*}
\end{lemma}
\begin{proof} 

From linearity of the code, 
\begin{equation*}
    F^{(\cC)}(\Gamma,z) = 
\begin{cases}
            F^{(\cC)}(\Gamma,0)  & \mbox{if } z|_{\Gamma} \mbox{ is a valid non-erasure pattern}\\
            0 & \mbox{otherwise.} 
	\end{cases}.
\end{equation*}
Furthermore, the number of distinct $z\in\cH_n$ for which $\cC(\Gamma,z)$ is nonempty equals 
  $
  2^{n-|\Gamma|}|\cC/\cC(\Gamma,0)|.
  $ 
  Hence,
    \begin{align*}
        \|F^{(\cC)}(\Gamma,\cdot)\|_{\alpha}  &= \left[ \frac{1}{2^n} \sum_{x\in \cH_n} F^{(\cC)}(\Gamma,x|_\Gamma)^{\alpha} \right]^{1/{\alpha} }\\
        & = \left[ \frac{1}{2^n} \frac{2^n |\cC|}{2^{|\Gamma|}} \frac{1}{F^{(\cC)}(\Gamma,0)} F^{(\cC)}(\Gamma,0)^{\alpha} \right]^{1/{\alpha} }\\
        & = \left[  \frac{|\cC|}{2^{|\Gamma|}} F^{(\cC)}(\Gamma,0)^{\alpha -1}\right]^{1/{\alpha} }.
    \end{align*}
\end{proof}
\begin{lemma}\label{lemma: Collision function - LHS of 3.4}
Let $\mathbbm{E}(f|\Gamma)$ be defined as in \eqref{eq:E}. Then,
    \begin{align*}
     \|\mathbbm{E}(2^nf_{\cC}|\Gamma)\|_{\alpha}  
         &=  \left[ \frac{2^{|\Gamma|}}{|\cC|} F^{(\cC)}_T(\Gamma,0) \right]^{(\alpha -1)/{\alpha} }.\\        
    \end{align*}
\end{lemma}

\begin{proof}
Using $f = 2^nf_{\cC}$ in \eqref{eq:E}, we obtain
\begin{align*}
    \mathbbm{E}(2^nf_{\cC}|\Gamma)(x) &= \frac{2^n}{2^{n- |\Gamma|}|\cC|} \sum_{y\in \cC : y|_\Gamma = x|_\Gamma} 1=  \frac{2^{|\Gamma|}}{|\cC|} F^{(\cC)}(\Gamma,x|_\Gamma),
\end{align*} 
and thus from Lemma \ref{lemma: Collision function- norms},
 \begin{align*}
        \|\mathbbm{E}(2^nf_{\cC}|\Gamma)\|_{\alpha}  &= \frac{2^{|\Gamma|}}{|\cC|}\|F^{(\cC)}(\Gamma,\cdot)\|_{\alpha} \\
        &= \frac{2^{|\Gamma|}}{|\cC|} \left[ \frac{|\cC|}{2^{|\Gamma|}} F^{(\cC)}(\Gamma,0)^{\alpha -1} \right]^{1/{\alpha} } \\
         &=  \left[ \frac{2^{|\Gamma|}}{|\cC|} F^{(\cC)}_T(\Gamma,0) \right]^{(\alpha -1)/{\alpha} }.
\end{align*}   
\end{proof}

\begin{lemma}\label{lemma: Collision function - RHS of 3.4}
    Let $\cC$ be a linear code. For $X = X_{\cC^{\bot}}, Y = Y_{(X,{\text{\rm BEC}}(\lambda))}$,
    \begin{align}\label{eq:HXY}
        H(X|Y) = \mathbbm{E}_{\Gamma\sim \lambda}\Big[\log \Big(\frac{2^{|\Gamma|}}{{|\cC|}} F^{\cC}(\Gamma,0)\Big) \Big].
    \end{align}  
\end{lemma}

\begin{proof} Start with taking $X = X_{\cC}$ and $ Y = Y_{({\text{\rm BEC}}(\lambda), X_{\cC})}$, then
    \begin{align*}
        H(X|Y=y) = \log[F^{(\cC)}(y)].
    \end{align*}
Therefore,
    \begin{align*}
        H(X|Y) & = \mathbbm{E}_Y[\log(F^{(\cC)}(Y))]\\
        & = \mathbbm{E}_{\Gamma}\mathbbm{E}_{Z|\Gamma}[\log(F^{(\cC)}(\Gamma,Z))|\Gamma]\\
         & = \mathbbm{E}_{\Gamma\sim 1-\lambda}[\log F^{(\cC)}(\Gamma,0)].
         \end{align*}
By a standard identity about dual matroids \cite[p.72]{Oxley1992},
 $$
        \dim(\cC^{\Gamma^c}) = \dim(\cC) - |\Gamma| + \dim(({\cC^\bot})^{\Gamma}),
 $$
or
  $$
  F^{(\cC)}(\Gamma,0) = \frac{|\cC|}{2^{|\Gamma|}} F^{\cC^{\bot}}(\Gamma^{c},0),
  $$
and thus we continue as follows:  
        \begin{align*}
         H(X|Y) & =  \mathbbm{E}_{\Gamma\sim 1-\lambda}\Big[\log\Big(\frac{|\cC|}{2^{|\Gamma|}} F^{\cC^\bot}(\Gamma^c,0)\Big)\Big]\\
        & = \mathbbm{E}_{\Gamma^c\sim \lambda}\Big[\log\Big(\frac{2^{|\Gamma^c|}}{{|\cC^\bot|}} F^{\cC^\bot}(\Gamma^c,0)\Big)\Big] \\
        & = \mathbbm{E}_{\Gamma\sim \lambda}\Big[\log\Big(\frac{2^{|\Gamma|}}{{|\cC^\bot|}} F^{\cC^\bot}(\Gamma,0)\Big)\Big].
    \end{align*}
Switching to the dual code and taking $X = X_{\cC^{\bot}} $ and $ Y = Y_{({\text{\rm BEC}}(\lambda), X_{\cC^{\bot}})}$ now yields
\eqref{eq:HXY}.
\end{proof}

\begin{lemma}\label{lemma: MacW}
   \begin{align*}
    \frac{\alpha}{\alpha -1} \mathbbm{E}_{\Gamma\sim \lambda}\log\|\mathbbm{E}(2^nf_{\cC}|\Gamma)\|_\alpha  =  \mathbbm{E}_{\Gamma\sim \lambda} \Ent[\mathbbm{E}(2^nf_{\cC}|\Gamma)] 
    = H(X_{\cC^{\bot}}|Y_{({\text{\rm BEC}}(\lambda), X_{\cC^{\bot}})}). \qedhere
\end{align*}  
\end{lemma}

\begin{proof}
 From Lemmas \ref{lemma: Collision function - LHS of 3.4} and \ref{lemma: Collision function - RHS of 3.4}, 
\begin{align*}
    \frac{\alpha}{\alpha -1} \mathbbm{E}_{\Gamma\sim \lambda}\log\|\mathbbm{E}(2^nf_{\cC}|\Gamma)\|_\alpha
    &= \mathbbm{E}_{\Gamma\sim \lambda}\Big[\log \Big(\frac{2^{|\Gamma|}}{{|\cC|}} F^{\cC}(\Gamma,0)\Big) \Big]\\ 
    &= H(X_{\cC^{\bot}}|Y_{({\text{\rm BEC}}(\lambda), X_{\cC^{\bot}})}),
\end{align*} 
which establishes the equality between the first and the third quantities. Since the second quantity is a limiting case of the first quantity and the value of the first quantity is independent of $\alpha$, we have equality between the first and the second quantities.
\end{proof}
Now Lemma \ref{lemma: smoothing and erasure} follows by
combining Lemma \ref{lemma: MacW} with Theorems \ref{thm: Samorodnitsky's inequality for entropy} and \ref{thm: Samorodnitsky's inequality for norms}.
\bibliographystyle{abbrvurl}
\bibliography{smoothing}

\end{document}